\documentclass[runningheads]{llncs}
\usepackage[T1]{fontenc}
\usepackage{graphicx}
\usepackage{hyperref}
\usepackage{color}

\usepackage{amsfonts}
\usepackage{mathtools}
\usepackage{comment}
\usepackage[style=alphabetic,natbib=true,maxcitenames=3]{biblatex}

\usepackage{tikz}
\usetikzlibrary{arrows}
\usepackage{float}

\usepackage{comment}

\newcommand{\e}{\varepsilon}
\newcommand{\N}{\mathbb{N}}
\newcommand{\Z}{\mathbb{Z}}
\newcommand{\R}{\mathbb{R}}
\DeclareMathOperator{\supp}{supp}
\DeclareMathOperator*{\argmax}{arg\,max}
\DeclareMathOperator*{\argmin}{arg\,min}
\DeclareMathOperator*{\E}{\mathbb{E}}

\DeclareMathOperator{\convhull}{conv}
\DeclareMathOperator{\prefix}{\mathcal{S}}
\DeclarePairedDelimiter{\iop}{(}{)}

\DeclarePairedDelimiter{\ios}{\{}{\}}
\DeclarePairedDelimiter{\absolute}{|}{|}
\DeclarePairedDelimiter{\normm}{\lVert}{\rVert}

\newcommand{\norm}{\normm*}
\newcommand{\p}{\iop*}
\newcommand{\s}{\ios*}
\newcommand{\abs}{\absolute*}

\newcommand{\eqdef}{\stackrel{\text{def}}{=}}
\newcommand{\emn}[2]{\mathcal{E}_{#1,#2}}
\newcommand{\pmax}{p_{max}}
\newcommand{\pmin}{p_{min}}

\begin{document}
\title{Correlated vs. Uncorrelated Randomness in Adversarial Congestion Team Games}
\titlerunning{Correlated vs. Uncorrelated Randomness in Congestion Games}
%
\author{Edan Orzech \and Martin Rinard}
%
\authorrunning{E. Orzech and M. Rinard}
%
\institute{MIT, Cambridge MA 02139, USA\\
\email{\{iorzech,rinard\}@csail.mit.edu}}
%
\maketitle              

\begin{abstract}
We consider team zero-sum network congestion games with $n$ agents playing against $k$ interceptors over a graph $G$. The agents aim to minimize their collective cost of sending traffic over paths in $G$, which is an aggregation of edge costs, while the interceptors aim to maximize the collective cost by increasing some of these edge costs. To evade the interceptors, the agents will usually use randomized strategies. We consider two cases, the correlated case when agents have access to a shared source of randomness, and the uncorrelated case, when each agent has access to only its own source of randomness. We study 
the additional cost that uncorrelated agents have to bear, specifically by comparing the costs incurred by agents in cost-minimal Nash equilibria when agents can and cannot share randomness.

We consider two natural cost functions on the agents, which measure the invested energy and time, respectively. We prove that for both of these cost functions, the ratio of uncorrelated cost to correlated cost at equilibrium is $O(\min(m_c(G),n))$, where $m_c(G)$ is the mincut size of $G$. This bound is much smaller than the most general case, where a tight, exponential bound of $\Theta((m_c(G))^{n-1})$ on the ratio is known.
We also introduce a set of simple agent strategies which are approximately optimal agent strategies. We then establish conditions for when these strategies are optimal agent strategies for each cost function, showing an inherent difference between the two cost functions we study.

\end{abstract}
\section{Introduction}
Network congestion games are central to the field of game theory \cite{koutsoupias1999worst,rosenthal1973class,roughgarden2002bad}. Recently, there has been work on team network congestion games, where a team of agents cooperate in a network congestion game against an adversary \cite{harks2022multi,bilo2023uniform}. Team network congestion games are motivated by the problem of fare evasion \cite{correa2017fare,harks2022multi}.
There, agents attempt to reach their destination while minimizing the amount of fare they pay, while some central authority attempts to counter this phenomenon with a policy that minimizes fare evasion on the network.
These games are also motivated by security applications (see \cite{sinha2018stackelberg} for a survey).
There, the agents attempt to attack a set of targets (for example entering a train station via a set of possible entrances), while a team of interceptors attempts to defend those targets with limited resources (for example by conducting inspections).

We consider a team zero-sum network congestion game in which a team of agents sends traffic
on a network of roads between two locations. 
Traversing each road segment incurs a cost, with the team attempting to minimize an aggregation of the segment costs. We consider two versions: the SUM version (Section~\ref{sec:sum}) and the MAX version (Section~\ref{sec:max}). In the SUM version, the agents minimize the total traversal cost, called the SUM cost. In the MAX version, the agents minimize the total traversal cost, called the MAX cost. Intuitively, the SUM cost measures the amount of total expended resources (such as transportation energy or tolls), while the MAX cost measures the amount of time required to transport the agents between the two locations.
On the other side, a team of interceptors acts against the agents by imposing costs on selected road segments, causing the agents to expend additional resources to traverse these segments. 

We model the road network with a directed acyclic graph $G$, possibly with multiple edges between a pair of vertices, with a source $s$ and a sink $t$. The edges model roads, the vertices $s,t$ model the origin and destination, while the remaining vertices are intersections. There is a team of $n$ agents and a team of $k$ interceptors playing on $G$. Each agent selects an $s-t$ path in $G$; each interceptor selects an edge of the graph $G$. Traversing an edge imposes a one time unit cost of 1, with an additional cost of $\alpha$ on each edge that an interceptor selects. Selected edge costs are additive; i.e., if two interceptors select the same edge, the additional cost is $2\alpha$.
By this we model settings where the interceptors can increase the cost at road segments, either by collecting additional tolls or by increasing congestion.

As usual, the agents and interceptors can use randomization to mix their strategies. We consider two simple communication assumptions between the agents. In the first case, agents are unable to communicate any random bits to each other (i.e., they are uncorrelated). In the second case, they can (i.e., they can be correlated). These two cases model scenarios where the interceptors may or may not be able to eavesdrop on the agents' communication. These cases also model scenarios where the agents' ability to communicate during the actual play is either disabled or not. We allow the interceptors to communicate freely and securely during the game. Therefore, they can be seen as one interceptor. However, for the sake of presentation, we choose to treat them as separate interceptors who make single and coordinated actions, rather than once interceptor who makes multiple actions.

To understand the impact of different (team) randomization capabilities, we measure the {\it price of uncorrelation} (PoU) of the game, an inefficiency index that was first introduced by \citet{basilico2017team}. There, the authors define PoU to be the ratio of the cost at a team-maxmin equilibrium (TME) of the game, to the cost at a correlated TME (TMECor) of the game. A TME (\cite{von1997team}) is a maxmin equilibrium between a team of uncorrelated agents and a single adversary, while in a TMECor the team agents can be correlated (see Section~\ref{sec:lit} for more details). The PoU of the model characterizes the inapproximability (in terms of cost) of correlated agent strategies using uncorrelated strategies: If the PoU is $r$, then the best approximation factor of any algorithm, regardless of time or space complexity, is exactly $\frac{1}{r}$ (this follows from the definition).

\citet{basilico2017team} show that in general, normal-form games, when the $n$ team agents have $m$ available strategies each,
they can incur an exponential PoU of $m^{n-1}$. As we discuss next, in our model the PoU is linear. This means that in the settings that we model, correlated randomness and communication in general between the agents is less valuable than in the games considered by \citet{basilico2017team}. Intuitively, the reason is as follows: We model scenarios where interceptors can increase the cost at road segments by only a finite amount (but not completely deny agents the ability to use the road segment).
On the other hand, \citet{basilico2017team} model scenarios where interceptors can completely prevent agents from traversing the graph (if an agent attempts to traverse a road segment that was essentially removed by an interceptor).
\\\\
{\bf Illustrating example\ }
We illustrate the PoU in our game in the following example.

\begin{example}
Consider the graph in Figure \ref{fig1}, with $m$ disjoint paths between $s$ and $t$ of the same length.
\begin{figure}[H]
\centering
\begin{tikzpicture}
    \tikzset{vertex/.style = {shape=circle,draw,minimum size=2.5em}}
    \tikzset{edge/.style = {->,> = latex'}}
    \node[vertex] (s) at (0,0) {$s$};
    \node[vertex,draw=none] (a) at (4,0) {$\vdots$};
    \node[vertex] (t) at (8,0) {$t$};
    \draw[edge] [bend left=35] (s) to node{} (t);
    \draw[edge] [bend left=15] (s) to node{} (t);
    \draw[edge] [bend right=20] (s) to node{} (t);
\end{tikzpicture}
\caption{Illustrating example}
\label{fig1}
\end{figure}
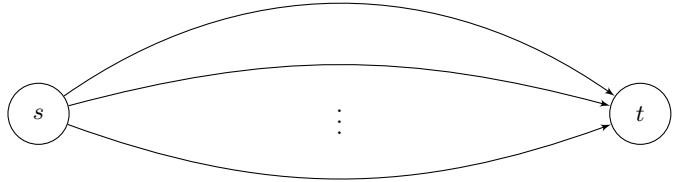
Suppose there are $n$ agents and one interceptor, playing with respect to the SUM cost function. Fix $\alpha>0$. When agents are correlated, the resulting equilibrium is one where agents pick together a path that is sampled from the $n$ paths of the graph uniformly at random.
The interceptor also picks a path uniformly at random. Then, the expected cost to the agents is $1+\frac{\alpha}{m}$. Now suppose that the agents are uncorrelated. Then one can show that the optimal strategy for the agents is having each agent independently pick a path uniformly at random, and the interceptor's strategy is unchanged. Then, the expected cost to the agents is $m\p{1-\p{1-\frac{1}{m}}^n}\p{1+\frac{\alpha}{m}}$. So the PoU is $m\p{1-\p{1-\frac{1}{m}}^n}=\Theta(\min(m,n))$. We note that in \cite{basilico2017team}, a game is considered which is played on the graph in Figure \ref{fig1}, but there the PoU is $m^{n-1}$.
\end{example}

\subsection{Our results}\label{subsec:results}

For both SUM and MAX costs, we establish bounds on the PoU (Sections \ref{sec:sum} and \ref{sec:max}), summarized as follows:

\begin{theorem}[Informal version of Theorem \ref{thm:general costs r bounds} and Corollary \ref{cor:general costs r bounds}]\label{thm:informal pou bounds}
The PoU, as a function of $\alpha$, denoted $r(\alpha)$, satisfies $r(\alpha)=\Theta(\min(m_c(G),n))$, where $m_c(G)$ is the size of a minimum cut of $G$.
\end{theorem}

This implies that in our games there is a polynomial-time $\Theta\p{\frac{1}{\min(m_c(G),n)}}$-approximation algorithm for playing a TMECor agent strategy using uncorrelated strategies. This is much smaller compared to the exponential (in agents) factor proved in \cite{basilico2017team} in their abstract class of games.

Next, we focus on analyzing more carefully the TME strategies of our game with the two cost functions.
In Section \ref{subsec:sum disjpaths}, we show that there is a set of simple strategies that admits optimal agent strategies over graphs with disjoint $s-t$ paths with respect to the SUM cost function.
In those strategies, agents roughly mix strategies uniformly over disjoint paths. Call these strategies {\it uniform strategies}, and denote them with $\mathcal{U}_n$ (see Section \ref{sec:defs} for the formal definition).

\begin{theorem}[Informal version of Theorem \ref{thm:min F disjoint paths}]\label{thm:SUM char}
For every graph $G$ where the $s-t$ paths are disjoint, and for every $n\ge1$, $\mathcal{U}_n$ admits an optimal agent strategy for $n$ uncorrelated agents.
\end{theorem}

In contrast, for the MAX cost, and for a larger set of graphs, we prove in Section \ref{subsec:max disjpaths} that the uniform strategies are almost never optimal.

\begin{theorem}[Informal version of Theorems \ref{thm:max uniform is almost never optimal general} and \ref{thm:max uniform equalengths is optimal}]\label{thm:MAX char}
\begin{enumerate}
    \item 
For a graph with disjoint $s-t$ paths, and for every $n>1$, $\mathcal{U}_n$ never admits an optimal agent strategy for $n$ uncorrelated agents if and only if the paths are not all of the same length, and $\alpha$ is large enough.

    \item
Let $G$ be a graph where, for $s-t$ every path there is a disjoint $s-t$ path of a different length. Then for every $n>1$, and for a large enough $\alpha>0$, the strategy set $\mathcal{U}_n$ never admits an optimal agent strategy for $n$ uncorrelated agents.
\end{enumerate}
\end{theorem}

Going back to the SUM cost, we define a notion of {\it uniform PoU}, denoted with $r_U$ where the agents can only use strategies in $\mathcal{U}_n$. Here, we prove appealing properties of $r_U$:

\begin{theorem}[Informal version of Theorem \ref{thm:r_U characterization} and Corollary \ref{coro:computing r_U}]\label{thm:r_U informal}
The function $r_U(\alpha)$ is a smooth, increasing function that for a large enough $\alpha$ satisfies $r_U(\alpha)=\Theta(\min(m_c(G),n))$. $r_U(\alpha)$ can be computed in polynomial time in the graph size.
\end{theorem}

Finally, in Section \ref{sec:known and slider} we formally connect our model with the model of \citet{basilico2017team}. The first result is a lemma where we show that cost functions can be converted to payoff functions while preserving the PoU:
\begin{lemma}[Informal version of Lemma \ref{lmm:c to u}]
For every game instance with a payoff function $u$, where the agents can guarantee positive payoff, there is a cost function $C$ such that $r(C)=r(u)$, and vice versa.
\end{lemma}
This implies that the PoU can be defined with either a cost function or a payoff function, thus establishing that our PoU bounds are indeed comparable to those of \cite{basilico2017team}.

\subsection{Our techniques}\label{subsec:techniques}
We now present and discuss the techniques we use in the proofs. Our techniques are novel in the context of the PoU.

\paragraph{Representing graphs as flow networks}

We define the model formally using terminology from flow networks. Sets of disjoint paths are identified with 0-1 flows from $s$ to $t$.
This view of the graph allows us to gain insights on the behavior of rational agents in our model. In short, for a given $\alpha$, they tend to ignore prohibitively long paths, and spread out over the remaining paths as much as possible. See Sections \ref{sec:sum} and \ref{sec:max} for further discussion.

Finally, we use this perspective on the problem to also prove the computational result of Theorem \ref{thm:r_U informal}. There, we translate a given game instance to a mincost flow problem instance (see \cite{ahuja1988network} for reference).

\paragraph{Approximating the PoU with a simplified PoU}
To prove
the $\Theta(\min(m_c(G),n))$ bound on the PoU,
we approximate the PoU with the uniform PoU. The uniform PoU can be seen as a simplified version of the PoU, a version that is simpler both to compute (for example Theorem \ref{thm:r_U informal} for SUM cost) and to understand. The second step of the technique is therefore to bound the uniform PoU. Put together, this two-step technique allows us to obtain asymptotically tight bounds on the PoU, while needing to consider only simple strategies in our calculations. Specifically, this technique enables us to prove (Theorem~\ref{thm:general cost ru is approx}) that, for every $n \geq 1$, and for sufficiently large $\alpha$:
\[
\p{1-O\p{\frac{1}{\alpha}}}r(\alpha)\le r_U(\alpha)\le\p{1+O\p{\frac{1}{\alpha}}}r(\alpha)
\]

Next, we discuss the techniques used in our equilibrium strategies analyses.

\paragraph{Reduction of strategy space}
In a search for optimal strategies, the simple forms of SUM and MAX allow us to simplify significantly the strategy spaces of both the agents and the interceptors. When simplifying the strategy spaces, we always keep at least one optimal strategy in the strategy space.

In Section \ref{subsec:sum disjpaths} we focus on the SUM cost over graphs of disjoint paths. There, the strategy space reduction technique is particularly effective:

\begin{lemma}[Informal version of Lemma \ref{coro:sum disjpaths opt in convhull}]\label{thm:sum disjpaths opt in convhull}
For every graph $G$ of disjoint paths, and for every $n\ge1$, the convex hull of $\mathcal{U}_n$ admits an optimal agent strategy for $n$ uncorrelated agents.
\end{lemma}

In the case of the MAX cost, this technique's effectiveness is more limited (since MAX is determined only by the longest path, rather than all paths).

\begin{lemma}[Informal version of Lemma \ref{lmm:assume iid max}]
For every graph $G$, and for every $n\ge1$, there exists an optimal agent strategy for $n$ uncorrelated agents where all agents use the same individual mixed strategy.
\end{lemma}
In other words, the agents' mixed strategies are distributed iid.
This produces optimal strategies which may seem counter-intuitive, where agents might intentionally send traffic on inherently long paths, while ignoring shorter paths that are also available.

\subsection{Organization}
In Section \ref{sec:lit} we overview related work.
In Section \ref{sec:defs} we formally define the problem and state some general lemmas.
In Section \ref{sec:sum} we analyze and prove our results for the SUM cost function.
In Section \ref{sec:max} we analyze and prove our results for the MAX cost function.
In the Appendices we provide omitted details proofs for some of the theorems and lemmas in this paper, and connect the results of \citet{basilico2017team} to ours as described above.

\section{Related Work}\label{sec:lit}
\paragraph{Team-maxmin equilibria} \citet{von1997team} study the maxmin value of zero-sum, normal-form games between a team and an adversary, where the team agents aim to maximize their identical payoffs. The team agents are uncorrelated. Von Stengel and Koller show that every team (mixed) strategy that attains the maxmin value of such a game is a part of a NE of the game, and in particular a payoff-maximal NE. Such a team strategy is called a {\it team-maxmin strategy}, and such a NE is called a {\it team-maxmin equilibrium} (TME). The payoff the team gets from TMEs is unique.
\citet{von1997team} introduce the concept of TME and only consider the case of uncorrelated agents. In contrast, we use the TME notion to model the impact that correlation between the agents has on their performance in our games, and prove results that highlight the extent to which correlation is significant in the games that we study.

Some work on determining TMEs and their values has been done in {\it rendezvous games} \cite{alpern1995rendezvous}, where two agents seek to meet at a point in space. \citet{lim1997rendezvous,lim2000rendezvous,alpern1998symmetric} consider versions of the rendezvous game where the agents in addition attempt to avoid being captured by an enemy searcher. \citet{lim1997rendezvous,lim2000rendezvous} consider correlated team agents playing against an enemy searcher on a clique and a circle (respectively), while \citet{alpern1998symmetric} consider `symmetric' team agents who cannot correlate their strategies if they want to meet. While these rendezvous games are multi-round games, these games have only two agents and one enemy and are played on restricted types of graphs. In comparison, while our games are single-round games, we allow any number of players on both teams, as well as any $s-t$ graph. This means that our paper studies the PoU for a much larger set of scenarios.

\paragraph{Computing TMEs} Computational aspects of TMEs have been thoroughly studied, showing both algorithms for finding TMEs (in both normal-form and extensive form games) \cite{zhang2020computing,kalogiannis2022efficiently,celli2018computational,basilico2017team,anagnostides2023algorithms} and hardness results \cite{anagnostides2023algorithms,hansen2008approximability,kalogiannis2022towards,celli2018computational}, in particular showing that computing a TME is FNP-hard \cite{hansen2008approximability,basilico2017team}.
In the computational aspect, we show (Corollary \ref{coro:computing r_U}) that for SUM cost, the PoU can be approximated in polynomial time.
Our focus in this paper is on asymptotic quantification of the PoU rather than an exact quantification. Note that the results mentioned in this paragraph imply that exact quantification is potentially a computationally hard problem, therefore there might not be an explicit and exact expression for the PoU in our games.

\paragraph{TMECor and PoU} When the team agents are correlated, they can be seen as one player, and then the game reduces to a 2-player zero-sum game. Then the TME is often called {\it TMECor} \cite{basilico2017team}, and can be found in polynomial time via linear programming. As was stated earlier, The PoU is the ratio of the payoff (of the team agents) at a TMECor to the payoff at a TME, and it was studied in \cite{basilico2017team,celli2020coordination}.
\citet{celli2020coordination} considers the PoU for extensive-form games, and shows bounds on it that depend on the number of leaves in the tree. When the game is in extensive form, one can pick whether the team agents have shared randomness only at the root or at all nodes of the game. This gives rise to refined notions of correlation \cite{celli2020coordination,zhang2020computing}. In our games, since they are single-round games, the distinction between these correlation notions does not exist. On the other hand, the extensive-form games mentioned above are abstract, while ours are concrete and motivated by the concrete scenarios mentioned in the introduction.

The PoU was also studied in \cite{schulman2017duality} as a difference rather than a ratio. The setting there is even more general, having 2 teams (call them team $A$ and team $B$) playing against each other in a team zero-sum game.
The main focus is on the {\it defensive gap} of such games, which is the largest difference between the minmax and the maxmin values of a game of a fixed size, and can be phrased as a sum of additive PoUs.
\citet{schulman2017duality} show that the additive version of PoU for team $A$ can get as large as $(1-\text{Val}T)(1-\hat{m}_A^{1-k_A})$, where $\text{Val}T$ is the minmax value of the game when the members of each team are correlated, $k_A$ is the size of team $A$ and $\hat{m}_A$ is the geometric mean of the number of strategies of team $A$'s members. This aligns with the $m^{n-1}$ bound of \citet{basilico2017team} when converting the ratio form into a difference form.
As discussed in the introduction, we show a linear PoU for our class of games.

\paragraph{Uniform mixed equilibria} A somewhat similar model was defined in the recent \cite{bilo2023uniform}, who consider congestion games such that team agents can only play uniform strategies, and face an adversary who is able to corrupt edges. As we show, after we change the cost function, uniform distributions will not always be the best strategies for the team agents, hence deviating from \cite{bilo2023uniform}. Our model is different than \cite{bilo2023uniform} in two ways. First, in our model the interceptors can only `slow down' the agents rather than causing link failures. Therefore, the way we model the interceptors models different settings than the settings modeled by \citet{bilo2023uniform}. Second, in our model each agent can play an arbitrary mixed strategy, as opposed to only picking ``$\rho$-uniform strategies''. In other words, in our model the agents have richer sets of (mixed) strategies, which makes our model more general.

\section{Preliminaries}\label{sec:defs}
\subsection{Notations}
For a set $X$ and $n\in\N$, we denote with $X^n$ the cartesian product of $n$ copies of $X$ and $X^*=\bigcup_{n\in\N}X^n$ as usual.
For $X$ which is either a set or a tuple $(x_i)_{i\in I}$, let $\bigcup X=\bigcup_{x\in X}x$ in the former case, and $\bigcup X=\bigcup_{i\in I}x_i$ in the latter.
For a set $X$, let $\Delta X$ be the set of distributions over $X$. For a distribution $\sigma\in\Delta X$, let $\supp\sigma$ be the support of $\sigma$. For $n\in\N$, let $\sigma^n$ be the product distribution of $n$ copies of $\sigma$. If $\sigma\in(\Delta X)^n$ we also write $\sigma=(\sigma_1,...,\sigma_n)$. For an set or event $A$ we denote the indicator of $A$ by $1_A$ or $1\s{x\mid x\in A}$. For conciseness, by $\eqdef$ we mean ``equal by definition''.

Throughout this paper, we will use the continuity of several functions in the joint distributions of uncorrelated agents. Then, we treat these distributions $\sigma=(\sigma_1,...,\sigma_n)$ as vectors $\sigma\in\R^{mn}$. Now, for the continuity purposes above we use the $\ell_2$-norm to measure distances between these distributions.

\subsection{The model}
{\bf The graph\ }
In the game there is a directed $s-t$ graph $G=(V,E,s,t)$, where $V$ is the set of vertices, $s,t\in V$ are the source and the sink of $G$ and $E$ is the {\it multiset} of edges.
We assume that $G$ is acyclic and that all edges are oriented from $s$ to $t$, i.e., every edge $e\in E$ is contained in some directed $s-t$ path.
Let $P_{st}=\s{p_1,...,p_m}$ be the set of directed $s-t$ paths in $G$. Paths will be treated as sets of edges,
which can be disjoint. Here, edge-disjointness and vertex-disjointness are the same because if 2 paths share a vertex but no edges, then they imply the existence of paths which are not edge-disjoint. On the other hand, vertex-disjointness implies edge-disjointness. If the paths are disjoint, we will assume that $|p_1|\le...\le|p_m|$,
and denote $P_i=\s{p_1,...,p_i}$.
In general we denote with $\pmax,\pmin$ an arbitrary pair of longest and shortest paths in $P_{st}$. We also define $|p_0|=0$. For an edge $e\in E$, let $P_{st,e}=\s{p\in P_{st}\mid e\in p}$.
\\\\
{\bf The strategies\ }
There are $n$ agents and $k$ interceptors. The set of pure strategies of each agent is defined to be $P_{st}$, and that of each interceptor is $E$. Paths picked by the agents will be denoted by $a=(a_1,...,a_n)$, and edges picked by interceptors will be denoted by $d=(d_1,...,d_k)$.
We will sometimes treat $a$ as an integer vector $(|a_i|)_i\in\Z^n\subseteq\R^n$, particularly when computing its norm $\norm a$. Here we will consider only the $\ell_\infty$-norm: $\norm{a}_\infty:=\max_{i\in[n]}|a_i|$.

Joint mixed strategies of agents and interceptors are denoted by $\sigma\in\Delta((P_{st})^n)$ and $\tau\in\Delta(E^k)$ respectively, and $\sigma_j,\tau_j$ are the marginal distributions of single agents and interceptors. If $P_{st}=\s{p_1,...,p_m}$, we will also write sometimes $\sigma_{ij}=\sigma_j(p_i)$.
For $\sigma\in\Delta((P_{st})^n)$ and $p\in P_{st}$, let $\sigma(p)=\Pr_{a\sim\sigma}(p\in a)$. Note that for $n=1$ this definition and the original definition of $\sigma(p)$ coincide. Let $\sigma^E$ be a distribution on subsets of $E$, where $\sigma^E(e)=\Pr_{a\sim\sigma}(e\in\bigcup a)$. Then, let $\beta(\sigma)=\max_{e\in E}\sigma^E(e)$.

When the agents are correlated, their set of mixed strategies is $\Delta((P_{st})^n)$. When they don't, it is only $(\Delta(P_{st}))^n$. 
The interceptors can always share their randomness.
\\\\
{\bf The cost functions\ }
The cost function of the game is $C:(P_{st})^*\times E^*\times(0,\infty)\to[0,\infty]$, where $C(a,d,\alpha)$ is the cost to the agents in the corresponding outcome. More formally, $-C$ is the utility of every agent, while $C$ is the utility of every interceptor. 
$\alpha$ will be referred to as either additional cost or delay imposed on edges by interceptors. The two cost functions we will consider for a large part of the paper are the SUM cost:
\begin{align}\label{eq:sum}
\tag{SUM}
    C(a,d,\alpha)=\abs{\bigcup a}+\alpha\sum_{j\le k}1\s{d_j\in\bigcup a}
\end{align}
and the MAX cost:
\begin{align}\label{eq:max}
\tag{MAX}
    C(a,d,\alpha)&=\max_{i\in[n]}\p{|a_i|+\alpha\sum_{j\in[k]}1\s{d_j\in a_i}}
\end{align}
$C$ is extended to mixed strategies in the usual way, as an expectation over the outcomes of the mixed strategies.
\\\\
{\bf Price of uncorrelation\ }
We measure the gap in the performances of correlated and uncorrelated agents with the PoU. It is defined by the costs at TMEs and TMECors of the game, as defined in \cite{von1997team,basilico2017team}.\footnote{Technically, the PoU is defined by payoff-maximizing agents. See Appendix \ref{sec:known and slider} to see why we can define it with cost-minimizing agents without changing it.}
\begin{definition}
For a game with graph $G$, $n$ agents, $k$ (correlated) interceptors, a parameter $\alpha$ and a cost function $C$, the PoU of the game is defined by
\begin{align}\label{eq:r}
    r(G,n,k,C,\alpha):=\frac{\min_{\sigma\in(\Delta(P_{st}))^n}\max_{d\in E^k}C(\sigma,\tau,\alpha)}{\min_{\sigma\in\Delta((P_{st})^n)}\max_{d\in E^k}C(\sigma,\tau,\alpha)}
\end{align}
\end{definition}

We omit $C$ when it is clear from context. Observe that $r(G,n,k,C,\alpha)\ge1$ as long as $C>0$. Note that in the PoU definition, the interceptor pick edges deterministically. This is wlog for the purposes of computing the costs at TMEs and TMECors of the game (or TME and TMECor values). As shown in \cite{von1997team}, for every $\sigma$ that attains the TME value (subject to a best response $d\in E^k$), there exists a $\tau\in\Delta(E^k)$ such that $(\sigma,\tau)$ is a TME of the game. A similar result holds for the TMECor value.

For $d\in E^k$ let $f_d(\sigma,\alpha)=C(\sigma,d,\alpha)$, and let $F_n(\sigma,\alpha)=\max_{d\in E^k}f_d(\sigma,\alpha)$, where $\sigma\in\Delta((P_{st})^n)$. In words, $F_n(\sigma,\alpha)$ is the cost that $n$ agents will incur when playing the joint strategy $\sigma$ against $k$ interceptors with strength $\alpha$.
Further observe that for SUM and MAX, when the agents minimize their cost and are correlated, they will always pick the same path. So Equation \ref{eq:r} simplifies to
\begin{align}
    r(G,n,k,C,\alpha)=\frac{\min_{\sigma\in(\Delta(P_{st}))^n}F_n(\sigma,\alpha)}{\min_{\sigma\in\Delta(P_{st})}F_1(\sigma,\alpha)}
\end{align}

We now define precisely the uniform strategies and the uniform PoU, using terminology from flow networks.
\\\\
{\bf $G$ as a flow network\ }
We recall the notion of mincuts. Endow every edge in $G$ with capacity 1.
An $s-t$ cut in $G$ is a partition of $V$ into $S\sqcup T$ such that $s\in S,~t\in T$. Let $E(S,T):=E\cap(S\times T)$ and $|(S,T)|=|E(S\times T)|$, and let $m_c(G)=\min_{S,T}|(S,T)|$, the size of a mincut of $G$. We write $m_c$ in short when $G$ is known from context.

We rephrase $m_c$ in terms of maximal integer flows (integer flows which cannot be extended). Let $MF(G)$ be the set of maximal integer (here, binary) $s-t$ flows on $G$. Again we write $MF$ when $G$ is known from context. By the maxflow-mincut theorem, $m_c=\max_{f\in MF}|f|$.\footnote{$G$ can be converted to a simple graph (where $E$ is not a multiset), and then the theorem holds, and when going back to the multigraph the fluxes and the cut sizes remain the same.}
We distinguish between maximal flows so that we have a mapping from each such flow to a unique maximal subset of edge-disjoint paths. $MF$ changes but $m_c$ and $\max_{f\in MF}|f|$ remain the same.

For an $s-t$ flow $f$, let $P_f=\s{p_1,...,p_{m'}}$ be the set of edge-disjoint paths on which a flow of 1 is transferred. Suppose that $|p_1|\le...\le|p_{m'}|$, and let $\prefix(P_f)=\s{\s{p_1,...,p_i}\mid i\le m'}$.
For $P$ let $U_P$ be the uniform distribution on $P$.
We now formally define the uniform strategies from earlier.

\begin{definition}
For $n$ uncorrelated agents, the set of uniform agent strategies is defined as
\begin{align}
    \mathcal{U}_n(G)=\s{U_P^n\mid\exists f\in MF,~P\in\prefix(P_f)}
\end{align}
\end{definition}

Each distribution $U_P^n$ is made up of $n$ independent copies of $U_P$, the uniform distribution on the set of paths $P$.
Now we can formally define the uniform PoU:

\begin{definition}
For a game with graph $G$, $n$ agents, $k$ (correlated) interceptors, a parameter $\alpha$ and a cost function $C$, the uniform PoU of the game is defined by
\begin{align}
    r_U(G,n,k,C,\alpha):=\frac{\min_{\sigma\in\mathcal{U}_n}F_n(\sigma,\alpha)}{\min_{\sigma\in\mathcal{U}_1}F_n(\sigma,\alpha)}
\end{align}
\end{definition}

We omit $C$ when clear from context. Note that $r_U\ge1$ is not generally true, however it is true in our problem instances, because the agents are correlated, even $\mathcal{U}_1$ will strictly dominate $\mathcal{U}_n$ in our problem instances.

\subsection{General lemmas}
Here we state a few general lemmas that are used in the proofs that appear here. For the full list of lemmas and their proofs, see Appendix \ref{sec:lemmas proofs}.

For $\ell,n\ge1$, let $\emn{\ell}{n}=\ell\p{1-\p{1-\frac{1}{\ell}}^n}$. In the appendix we show that when $\ell,n$ are large, $\emn{\ell}{n}=\Theta(\min(\ell,n))$, with an $\le$ inequality in general.

\begin{lemma}\label{lmm:beta lower bound}
Let $\sigma\in(\Delta(P_{st}))^n$. Let $G'$ result from $G$ by keeping only vertices and edges on paths $p$ for which $\sigma(p)>0$. Let $m'=m_c(G')$.
Then $\beta(\sigma)\ge1-\p{1-\frac{1}{m'}}^n=\frac{\emn{m'}{n}}{m'}$, and this bound is tight.
\end{lemma}

\begin{lemma}\label{lmm:beta upper bound n to 1}
For every $\sigma\in\Delta(P_{st})$ and $n\in\N$, $\beta(\sigma)\le\beta(\sigma^n)\le\min(m_c,n)\beta(\sigma)$.
\end{lemma}

\begin{lemma}\label{lmm:r upper bound n to 1}
Let $x\in\R$. Suppose that $C$ is $F$-nice and for every $\sigma\in\Delta(P_{st})$ there holds $F_n(\sigma^n)\le xF_1(\sigma)$. Then $r(G,n,k,\alpha)\le x$.
\end{lemma}

Suppose that $P_{st}=\s{p_1,...,p_m}$ is a set of disjoint paths. Let $S=\{X\in\Delta(P_{st})\mid\forall j<j',~~X(p_i)\ge X(p_{i'})\}$, the set of decreasing distributions with respect to the indexing of $P_{st}$. In our analysis, $S$ will come up in the case of disjoint paths, where path lengths are increasing with respect to indexing.

\begin{lemma}\label{lmm:convhull}
$S=\convhull(\mathcal{U}_1)$, and for every $n,$ $\convhull(\mathcal{U}_n)=\s{X^n\mid X\in\convhull(\mathcal{U}_1)}$.
\end{lemma}

\section{SUM Cost}\label{sec:sum}
The full proofs of this section can be found in Appendix \ref{appendix:sum}.

\subsection{General $s-t$ graphs}\label{subsec:sum general st graphs}
We begin with Theorem \ref{thm:informal pou bounds} for the SUM cost.
We begin with the following technical lemma. Essentially, it observes that the SUM cost is a sum of two factors. The first is the cost incurred by the interceptors, which is proportional to the maximum probability mass placed on an edge by the joint strategy of the agents. The second factor is simply the expected number of edges used by the agents.

\begin{lemma}\label{lmm:sum cost F form}
$F_n(\sigma,\alpha)=k\alpha\beta(\sigma)+\E_{a\sim\sigma}\p{\abs{\bigcup a}}$,
and when $P:=\bigcup_j\supp\sigma_j$ is a set of disjoint paths, $F_n(\sigma,\alpha)=k\alpha\beta(\sigma)+\sum_{p\in P}\sigma(p)|p|$.
\end{lemma}

We now state the following theorem about uniform strategies, which shows that they are approximately optimal agent strategies.

\begin{theorem}\label{thm:general cost uniform is approx}
Let $\e>0$, and $\alpha\ge\frac{m_c}{k\e\emn{m_c}{n}}\p{n\abs\pmax-(1+\e)\pmin}$. Then
\begin{align*}
    \min_{\sigma\in\mathcal{U}_n}F_n(\sigma,\alpha)\le(1+\e)\min_{\sigma\in(\Delta(P_{st}))^n}F_n(\sigma,\alpha)
\end{align*}
\end{theorem}

Intuitively, as $\alpha$ grows large, the paths' inherent lengths (namely, ignoring $\alpha$) become smaller than $\alpha$. Therefore, the agents will get increasingly concerned with avoiding encounters with interceptors over the base cost of sending traffic over long paths. To address this concern, the agents will prefer strategies $\sigma$ that have small $\beta(\sigma)$. This is because the additional cost caused by the interceptors will be $k\alpha\beta(\sigma)$, and the agents can only control the rightmost factor.

The theorem above implies our two main results. The first is the approximation of the PoU with the uniform PoU.

\begin{theorem}\label{thm:general cost ru is approx}
Let $\e>0$, and $\alpha\ge\frac{m_c}{k\e}\p{n\abs\pmax-(1+\e)\abs\pmin}$. Then
\begin{align*}
    \frac{r(G,n,k,\alpha)}{r_U(G,n,k,\alpha)}\in\left[\frac{1}{1+\e},1+\e\right]
\end{align*}
\end{theorem}

The second main result of this section is our bounds on the PoU. We provide both asymptotic and non-asymptotic bounds.

\begin{theorem}\label{thm:general costs r bounds}
Let $\e>0$. We have the following bounds on $r$:
\begin{enumerate}
    \item (Simple upper bound) $r(G,n,k,\alpha)\le\min(m_c,n)$.
    \item (Upper bound) Let
    $\alpha\ge\frac{m_c}{k\emn{m_c}{n}\e}\p{n\abs\pmax-(1+\e)\emn{m_c}{n}\abs\pmin}$.
    Then
    $$r(G,n,k,\alpha)\le(1+\e)\emn{m_c}{n}$$
    \item (Lower bound) Let $\alpha\ge\frac{m_c}{k\emn{m_c}{n}\e}\p{(1-\e)\emn{m_c}{n}n\abs\pmax-\abs\pmin}$. Then
    $$r(G,n,k,\alpha)\ge(1-\e)\emn{m_c}{n}$$
    \item (Limit) $\lim_{\alpha\to\infty}r(G,n,k,\alpha)=\emn{m_c}{n}$.
\end{enumerate}
\end{theorem}

We conjecture that for general $s-t$ graphs, there holds
\begin{align}
    \min_{\sigma\in(\Delta(P_{st}))^n}F_n(\sigma,\alpha)=\min_{\sigma\in\mathcal{U}_n}F_n(\sigma,\alpha)
\end{align}
and as a consequence, there holds $r(G,n,k,\text{SUM},\alpha)=r_U(G,n,k,\text{SUM},\alpha)$. This is a stronger result about SUM than Theorem \ref{thm:general costs r bounds}, for which we do not have a proof at this time. However, we do prove the conjecture in the case where $G$ has only disjoint paths, which seems to support the validity of the conjecture. If the conjecture is true, then $r$ will satisfy all the properties that $r_U$ satisfies.
Because SUM has convenient properties (e.g. Lemma \ref{lmm:sum cost F form}), $r_U$ also has several appealing properties: it is smooth a.s. (almost surely), continuous and strictly increasing in $\alpha$ until it hits $\emn{m_c}{n}$, where it plateaus. In addition, $r_U$ can be computed in polynomial time in the game size.

\subsubsection*{Properties of $r_U$}
For every $i\in[m_c]$ pick a subset$$P^*_i\in\argmin_{f\in MF,P\in\prefix(P_f),|P|=i}\sum_{p\in P}|p|$$
By definition, for every $n\in\N$
\begin{align*}
    \min_{U^n\in\mathcal{U}_n}F_n(U^n,\alpha)=\min_iF_n(U_{P^*_i}^n,\alpha)=\min_i\p{\frac{\emn{i}{n}}{i}\p{k\alpha+\sum_{p\in P^*_i}|p|}}
\end{align*}
This defines a more explicit form for $r_U$. By the minimality of each $P^*_i$, the sums $S_i:=\sum_{p\in P^*_i}|p|$ form a strictly increasing sequence. We now formally state the properties of $r_U$ described above.

\begin{theorem}\label{thm:r_U characterization}
$r_U(G,n,k,\alpha)$ is a function that is smooth a.s., continuous and strictly increasing in $\alpha$. Furthermore, there holds $r_U(G,n,k,\alpha)\le\emn{m_c}{n}$. There is a function $A(G,n,k)\in[0,\infty)$ such that equality holds for every $\alpha\ge A(G,n,k)$.    
\end{theorem}

We remark that since $k$ has the same effect on $F_n$ as $\alpha$, an analogous claim to the above holds as a function of $k$.

Next, we have a polynomial-time algorithm for computing $r_U$, using mincost flow algorithms (see \cite{ahuja1988network} for reference).

\begin{corollary}\label{coro:computing r_U}
$r_U(G,n,k,\alpha)$ can be computed in time $O\p{|E|^{1+o(1)}m_c}$ whp, and can be computed in time $O\p{|E|m_c|V|\log|V|}$ deterministically.
\end{corollary}

We also find bounds on the minimal $\alpha$ for which $r_U(G,n,k,\alpha)=\emn{m_c}{n}$. As we show, the lower bound can be exponential in terms of $m_c,n$ and the paths lengths in $G$.

\begin{theorem}\label{thm:alpha bound}
$A(G,n,k)\ge\frac{e^{\Theta(n/m_c)}\p{S_{m_c}-S_{m_c-1}}-S_{m_c}}{k}$.
\end{theorem}

\subsection{Disjoint paths}\label{subsec:sum disjpaths}
As stated earlier, our main result for the SUM cost and disjoint paths is that $\mathcal{U}_n$ always admits an optimal agents strategy. Hence $r=r_U$ in these instances. Intuitively, for fixed $G,k,\alpha$, some paths in $G$ are too expensive for the agents to consider -- they would rather incur $k\alpha$ additional cost than pick these long paths, so they ignore them. Then, $k\alpha$ becomes very large, relative to the remaining paths. Hence, the agents will pick a joint strategy that minimize the probability that agents collide with interceptors, and $\mathcal{U}_n$ offers one.

To prove the result, we reduce the strategy space that we need to consider. Next we will establish Lemma \ref{coro:sum disjpaths opt in convhull} in two steps. The first reduces the strategy space to distributions whose mass on paths decreases as the path gets longer.

\begin{lemma}\label{lmm:assume decreasing sum}
When minimizing $F_n(\cdot,\alpha)$ we can restrict ourselves to the set of decreasing distributions: all $\sigma$ such that for $i<i'$, $\sigma(p_i)\ge\sigma(p_{i'})$.
\end{lemma}

For such $\sigma$, $F_n(\sigma,\alpha)=\sum_{p_i\in P_{st}}\sigma(p)+\alpha\sigma(p_1)$. Our second reduction is to iid agent distributions.

\begin{lemma}\label{lmm:assume iid sum}
When minimizing $F_n(\cdot,\alpha)$ we can restrict ourselves to the set of iid distributions: all $\sigma$ which are the product distribution of $n$ copies of a distribution on $P_{st}$.
\end{lemma}

These two lemmas and Lemma \ref{lmm:convhull} imply the following lemma (which was discussed in Section \ref{subsec:techniques}).

\begin{lemma}\label{coro:sum disjpaths opt in convhull}
$\min_{\sigma\in(\Delta(P_{st}))^n}F_n(\sigma,\alpha)=\min_{\sigma\in\convhull(\mathcal{U}_n)}F_n(\sigma,\alpha)$.   
\end{lemma}

Now, define $F^1_n(X)=F_n(X^n,\alpha)$ for $X\in\convhull(\mathcal{U}_1)$. Our goal is to compute $\min_{X\in\convhull(\mathcal{U}_1)}F^1_n(X)$. 
$F^1_n$ is concave on $\convhull(\mathcal{U}_1)$ as a positive sum of concave functions: this is because each $\sigma(p_i)$ is a function of the form $x\mapsto1-(1-x)^n$. We now obtain the optimality of $\mathcal{U}_n$ with respect to the SUM cost, as was discussed in the Section \ref{subsec:results}.

\begin{theorem}\label{thm:min F disjoint paths}
Let $\alpha>0$. Then $\min_{\sigma\in(\Delta(P_{st}))^n}F_n(\sigma,\alpha)=\min_{\sigma\in\mathcal{U}_n}F_n(\sigma,\alpha)$.
Hence $r(G,n,k,\alpha)=r_U(G,n,k,\alpha)$.
\end{theorem}

\section{MAX Cost}\label{sec:max}
The full details of this section can be found in Appendix \ref{sec:max details}.
Throughout the rest of the paper, we will assume that $\alpha>\frac{\abs\pmax-\abs\pmin}{k}$. If $\alpha<\frac{\abs\pmax-\abs\pmin}{k}$, then the agents can strictly decrease their worst-case cost by never picking $\pmax$. Then, $P_{st}$ becomes effectively $P_{st}\setminus\{\pmax\}$. For simplicity, we do not consider the case $\alpha=\frac{\abs\pmax-\abs\pmin}{k}$ here.

In short, the main result here is that in nontrivial game instances with the MAX cost, $\mathcal{U}_n$ does not admit an optimal agent strategy (Theorems \ref{thm:max uniform is almost never optimal general} and \ref{thm:max uniform is almost never optimal}).

In Appendix \ref{subapp:max eg} we illustrate this with an example.

\subsection{General $s-t$ graphs}\label{subsec:max general}
Again, we have the function $F_n(\sigma,\alpha)=\max_{d\in E^k}f_d(\sigma,\alpha)$ for every $n\in\N$, however $F_n$ doesn't have the same form as in Section \ref{sec:sum} for $n>1$. However, $F_1$ does collide with the definition that results from taking SUM cost, which we have analyzed. Therefore, we only have $F_n$ for $n>1$ to analyze.

We begin by establishing Theorem \ref{thm:general cost uniform is approx}, Corollary \ref{thm:general cost ru is approx} and Theorem \ref{thm:general costs r bounds} from Section \ref{sec:sum} but for the MAX cost. The proofs are almost identical, regardless of the cost function. In order for these proofs to go through, it suffices to prove the following lemma:
\begin{lemma}\label{lmm:max Fn bounds}
For every $\sigma\in\Delta(P_{st})$,
\begin{align*}
    \abs\pmin+k\alpha\le F_n(\sigma,\alpha)\le\abs\pmax+k\alpha\le n\abs\pmax+k\alpha
\end{align*}
\end{lemma}

\begin{corollary}\label{cor:general costs r bounds}
Theorems \ref{thm:general cost uniform is approx}, \ref{thm:general cost ru is approx} and \ref{thm:general costs r bounds} all hold for the MAX cost verbatim.
\end{corollary}

An analogue to Lemma \ref{lmm:assume iid sum} holds for MAX as well, but for general $s-t$ graphs:

\begin{lemma}\label{lmm:assume iid max}
When minimizing $F_n(\cdot,\alpha)$ we can restrict ourselves to the set of iid distributions: all $\sigma$ which are the product distribution of $n$ copies of a distribution on $P_{st}$.
\end{lemma}

Before focusing on graphs of disjoint paths, we end with stating the second part of Theorem \ref{thm:MAX char}:

\begin{theorem}\label{thm:max uniform is almost never optimal general}
Let $G$ be an $s-t$ graph, and let $n>1,k\ge1$ such that the following holds: For every $f\in MF$, and for $P_f:=\s{p_1,...,p_{m'}}$, there holds $|p_1|<|p_{m'}|$. Also, $\alpha$ is large enough such that $\min_{\sigma\in(\Delta(P_f))^n}F_n(\sigma,\alpha)<|p_1|+\frac{\emn{m'}{n}}{m'}k\alpha$ for every such $f,P_f$.

Then
\begin{align*}
    \min_{\sigma\in(\Delta(P_{st}))^n}F_n(\sigma,\alpha)<\min_{\sigma\in\mathcal{U}_n}F_n(\sigma,\alpha)
\end{align*}
\end{theorem}

When $G$ has only disjoint paths, the theorem has a simpler phrasing. Also, the proof of the general case follows immediately from the disjoint-paths case (using the definition of $\mathcal{U}_n$), so we only prove the latter (Theorem \ref{thm:max uniform is almost never optimal}).

\subsection{Disjoint paths}\label{subsec:max disjpaths}
\subsubsection*{The general $k$ case}
As the paths are disjoint, we sometimes abuse notation by identifying each $d\in E^k$ with the corresponding tuple in $[m]^k$.

We begin with proving Theorem \ref{thm:max uniform is almost never optimal}. Fix $G$ with disjoint paths $P_{st}$, such that $|p_1|<|p_m|$. We will sometimes consider $G$ with only the first $i$ paths for $i\le m$, denote it as $G_i$. Then, the $F_n$ function defined on $G_i$ is simply $F_n$ when the domain of the distributions is $(\Delta(P_i))^n$ instead of $(\Delta(P_{st}))^n$ (because the interceptors will decrease the cost to the agents if they pick edges from paths $p_{i+1},...,p_m$, in contrast with their goal of maximizing the agents' cost). Let $m_1$ be the largest integer such that $|p_1|=...=|p_{m_1}|$.

Suppose the agents play only on $G_{m_1}$. All these paths have the same length, so intuitively the agents minimize their cost by minimizing their probability of picking a path with a large imposed delay. The best strategy for agents is then $U_{P_{m_1}}^n\in\mathcal{U}_n$. Intuitively, the interceptors will respond by picking $k$ edges from the same path (with some distribution over $P_{m_1}$), because in each choice of paths and edges, the cost is determined only by a costliest path. The cost to the agents is $|p_1|+\frac{\emn{m_1}{n}}{m_1}k\alpha$. Suppose that $\alpha$ is large so that $\min_{\sigma\in(\Delta(P_{st}))^n}F_n(\sigma,\alpha)<|p_1|+\frac{\emn{m_1}{n}}{m_1}k\alpha$. As we show next, this makes {\it all} strategies in $\mathcal{U}_n$ no longer optimal for the agents. Below we make these observations concrete.

We begin with the key lemma, which holds for any graph of disjoint paths.
\begin{lemma}\label{lmm:max Fn(U)=f_m^k(U)}
For every $G$ with $m>1$ disjoint paths, for every $n>1,k\ge1$ and $\alpha>\frac{|p_m|-|p_1|}{k}$: $F_n(U_{P_{st}}^n,\alpha)=f_{(m,...,m)}(U_{P_{st}}^n,\alpha)>\max_{d\ne(m,...,m)}f_d(U_{P_{st}}^n,\alpha)$.
\end{lemma}

We apply Lemma \ref{lmm:max Fn(U)=f_m^k(U)} to $F_n$ on a subset of $\mathcal{U}_n$, specifically those distributions which are supported on the first $m_1$ paths.

\begin{lemma}\label{lmm:max uniforms in equal lengths}
For every $i\le m_1$, $F_n(U_{P_i}^n,\alpha)=|p_1|+\frac{\emn{i}{n}}{i}k\alpha$. We also have $\min_iF_n(U_{P_i}^n,\alpha)=F_n(U_{P_{m_1}}^n,\alpha)$.
\end{lemma}

Next, an upper bound on the minimal $\alpha$ for which $U_{P_{m_1}}^n$ is no longer optimal for the agents. Let $E'_i=\frac{\emn{i}{n}}{i}$.

\begin{lemma}\label{lmm:bound for that}
Let $\alpha>\frac{(|p_{m_1+1}|-|p_1|)E'_{m_1+1}}{k(E'_{m_1}-E'_{m_1+1})}$. Then $\min_{\sigma\in(\Delta(P_{st}))^n}F_n(\sigma,\alpha)<|p_1|+\frac{\emn{m_1}{n}}{m_1}k\alpha$.
\end{lemma}

Now we can state the disjoint-paths version of Theorem \ref{thm:max uniform is almost never optimal general}:
\begin{theorem}\label{thm:max uniform is almost never optimal}
Let $G$ with $m>1$ disjoint paths where $|p_1|<|p_m|$. For every $n,k>1$ and for every $\alpha>\frac{|p_m|-|p_1|}{k}$, if $\min_{\sigma\in(\Delta(P_{st}))^n}F_n(\sigma,\alpha)<|p_1|+\frac{\emn{m_1}{n}}{m_1}k\alpha$, then $\min_{\sigma\in(\Delta(P_{st}))^n}F_n(\sigma,\alpha)<\min_{\sigma\in\mathcal{U}_n}F_n(\sigma,\alpha)$.
\end{theorem}

We end with finishing the first part of Theorem \ref{thm:MAX char}.

\begin{theorem}\label{thm:max uniform equalengths is optimal}
Let $G$ with $m$ disjoint paths where $|p_1|=|p_m|$. For every $n,k$ and for every $\alpha>0$:
\begin{align*}
    \min_{\sigma\in(\Delta(P_{st}))^n}F_n(\sigma,\alpha)=F_n(U_{P_{st}}^n,\alpha)=|p_1|+\frac{\emn{m}{n}}{m}k\alpha
\end{align*}
\end{theorem}

\subsubsection*{The $k=1$ case}
Here, we have an analogue of Lemma \ref{lmm:assume decreasing sum}:
\begin{lemma}\label{lmm:assume decreasing max}
When minimizing $F_n(\cdot,\alpha)$ we can restrict ourselves to the set of decreasing distributions: all $\sigma$ such that for $i<i'$, $\sigma(p_i)\ge\sigma(p_{i'})$.
\end{lemma}

Using Lemma \ref{lmm:convhull}, we get a corollary for MAX cost which is analogous to the SUM cost case:
\begin{corollary}
When $k=1$, there holds
\begin{align*}
    \min_{\sigma\in(\Delta(P_{st}))^n}F_n(\sigma,\alpha)=\min_{\sigma\in\convhull(\mathcal{U}_n)}F_n(\sigma,\alpha)
\end{align*}
\end{corollary}

But in contrast to the SUM cost case, Theorem \ref{thm:max uniform is almost never optimal} implies that in many game instances $F_n$ is not concave on $\mathcal{U}_n$:
\begin{corollary}
Let $G$ a graph with $m$ disjoint paths, and $m_1$ shortest paths. Let $k=1$ and $n>1$. Let $\alpha$ such that $\min_{\sigma\in(\Delta(P_{st}))^n}F_n(\sigma,\alpha)<|p_1|+\frac{\emn{m_1}{n}}{m_1}k\alpha$. Then $F_n$ is not concave on $\mathcal{U}_n$.
\end{corollary}

\begin{credits}
\subsubsection*{Acknowledgements.} We thank Farid Arthaud and Ziyu Zhang for helpful discussions and valuable feedback on earlier versions of the introduction. We also thank Yannai A. Gonczarowski for helpful discussions, and Kai Jia for valuable feedback on earlier versions of the introduction.
\end{credits}

\printbibliography{}

\newpage

\appendix

\section{Proofs of General Lemmas}\label{sec:lemmas proofs}
Recall that $\emn{\ell}{n}=\ell\p{1-\p{1-\frac{1}{\ell}}^n}$. We begin with basic properties of $\emn{\ell}{n}$.
\begin{lemma}
\begin{enumerate}
    \item When $\ell,n\ge1$, $\emn{\ell}{n}$ is strictly increasing in $\ell$ and strictly increasing in $n$. On the other hand, $\frac{\emn{\ell}{n}}{\ell}$ is strictly decreasing in $\ell$ and strictly increasing in $n$.
    \item Treat $n=n(\ell)$ as a function of $\ell$. Then
    \begin{align*}
        1=\begin{cases}
        \lim_{\ell\to\infty}\frac{\emn{\ell}{n}}{n}&n=o(\ell)\\
        \lim_{\ell\to\infty}\frac{\emn{\ell}{n}}{(1-e^{-x})\ell}&n=x\ell+o(\ell),~x\in(0,\infty)\\
        \lim_{\ell\to\infty}\frac{\emn{\ell}{n}}{\ell}&n=\omega(\ell)
        \end{cases}
    \end{align*}
    i.e., $\emn{\ell}{n}$ converges to the functions subtracted from it in the limits above, depending on whether $n=o(\ell)/\Theta(\ell)/\omega(\ell)$.
\end{enumerate}
\end{lemma}
\begin{proof}
Item 1 is easy to verify, so we prove only Item 2. Suppose $n=\omega(\ell)$. Then by elementary calculus, $\lim_{\ell\to\infty}\p{1-\frac1\ell}^n=0$. Therefore, $\lim_{\ell\to\infty}\frac{\emn{\ell}{n}}{\ell}=1$. Now suppose that $n=o(\ell)$. On the one hand, by Bernoulli's inequality $\emn{\ell}{n}\le\ell\p{1-\p{1-\frac{n}{\ell}}}=n$. On the other hand, because $\p{1-\frac{1}{\ell}}^n\le n\p{1-\frac{1}{l}}$ for $\ell,n\ge1$, we have $\emn{\ell}{n}\ge\ell\p{1-n\p{1-\frac{1}{\ell}}}=n-\frac{n}{\ell}$. Therefore
\begin{align*}
    1\ge\lim_{\ell\to\infty}\frac{\emn{\ell}{n}}{n}\ge\lim_{\ell\to\infty}\frac{n}{\frac{n}{\ell}-n}=1
\end{align*}
Now suppose that $n=x\ell+o(\ell)$. Then $\emn{\ell}{n}=\ell\p{1-\p{1-\frac{1}{\ell}}^{x\ell+o(\ell)}}$. Again by calculus, $\lim_{\ell\to\infty}\frac{\p{1-\frac{1}{\ell}}^{x\ell+o(\ell)}}{e^{-x}}=1$, as $x$ is constant. So $\lim_{\ell\to\infty}\frac{\emn{\ell}{n}}{(1-e^{-x})\ell}=1$.
\qed\end{proof}
This means that unless $n=x\ell+o(\ell)$, then $\emn{\ell}{n}=(1+o(1))\min(\ell,n)$, where $o(1)=o_{\ell,n}(1)$. 
If $n=x\ell+o(\ell)$, then $\emn{\ell}{n}=(1+o(1))(1-e^{-x})\ell\in\left[(1-e^{-1})\ell,\ell\right]$.

We now prove the lower bound $\beta$ of Lemma \ref{lmm:beta lower bound}. For that we have the following lemma:

\begin{lemma}\label{lmm:with amgm}
Let $p\in P\subseteq P_{st}$. The distribution $U_P^n$ is a solution of the maximization problem
\begin{align*}
    \max_{\sigma\in(\Delta(P_{st}))^n}~~~&\prod_j(1-\sigma_j(p))\\
    \text{subject to~~~}&\forall j\in[n]~~~\sum_{p'\in P}\sigma_j(p')=1\\
    &\forall p'\in P~~~\prod_j(1-\sigma_j(p))=\prod_j(1-\sigma_j(p'))
\end{align*}
\end{lemma}
\begin{proof}
By the AMGM inequality
\begin{align*}
    \sqrt[n]{\prod_j(1-\sigma_j(p))}\le\frac{n-\sum_j\sigma_j(p)}{n}
\end{align*}
By the second constraint we have
\begin{align*}
    \forall p'\in P~~~\sqrt[n]{\prod_j(1-\sigma_j(p))}\le\frac{n-\sum_j\sigma_j(p')}{n}
\end{align*}
and by the first constraint
\begin{align*}
    \sum_{p'\in P}\sum_j\sigma_j(p')=n~\Rightarrow~\exists p'~~~\sum_j\sigma_j(p')\ge\frac{n}{|P|}
\end{align*}
Therefore for that $p'$,
\begin{align*}
    \sqrt[n]{\prod_j(1-\sigma_j(p))}\le1-\frac{1}{|P|}\Rightarrow\prod_j(1-\sigma_j(p))\le\p{1-\frac{1}{|P|}}^n
\end{align*}
so $\p{1-\frac{1}{|P|}}^n$ is an upper bound on the maximal value of $\prod_j(1-\sigma_j(p))$, and $\sigma=U_P^n$ attains this bound.
\qed\end{proof}

\begin{proof}[Proof of Lemma \ref{lmm:beta lower bound}]
Let $E',P'$ be the sets of edges, $s-t$ paths of $G'$. Take a maximum flow $f$ of $G'$ with $|f|=|P_f|=m'$ and denote $P_f=\s{p_1,...,p_{m'}}$. For every $i$ pick an edge $e_i\in p_i$ such that $\s{e_1,...,e_{m'}}$ forms a mincut of $G'$ (which is possible because $f$ saturates some mincut). Then $\beta(\sigma)\ge\max_i\sigma^{E'}(e_i)$. $f$ is maximal and $\s{e_1,...,e_{m'}}$ intersects with every path in $P'$. So we argue that $\max_i\sigma^{E'}(e_i)$ is minimal when $\bigcup_j\supp\sigma_j\subseteq P_f$. This is because if $\sigma(p)>0$ for $p\notin P_f$, then for $e_i$ such that $e_i\in p$ we can move the weight of $\sigma$ on $p$ to $p_i$, and get the corresponding $\mu$. It satisfies $\mu^{E'}(e_i)\eqdef\sigma^{E'}(e_i)$, and $\mu^{E'}(e_{i'})=\sigma^{E'}(e_{i'})$ for $e_{i'}\notin p$, and $\mu^{E'}(e_{i'})<\sigma^{E'}(e_{i'})$ for $e_{i'}\in p$.
Therefore at the minimum, $\sigma^{E'}(e_i)=\sigma(p_i)=1-\prod_j(1-\sigma_j(p_i))$.

Now, suppose that $\min_i\sigma(p_i)<\max_i\sigma(p_i)$. Then by continuity we can slightly increase $\min_i\sigma(p_i)$ and slightly decrease each $\sigma(p_{i'})=\max_i\sigma(p_i)$, to get a product distribution $\sigma'$ with $\max_i\sigma'(p_i)<\max_i\sigma(p_i)$. Therefore, at the minimum we have $\sigma(p_i)=\sigma(p_{i'})$ for every $i,i'\le m'$. We finish the proof by applying Lemma \ref{lmm:with amgm} and with it evaluating $\beta$ at $\sigma=U_{P_f}^n$.
\qed\end{proof}

\begin{proof}[Proof of Lemma \ref{lmm:beta upper bound n to 1}]
Let $\sigma\in\Delta(P_{st})$. The inequality $\beta(\sigma^n)\le m_c\beta(\sigma)$ follows from $\beta(\sigma)\ge\frac{1}{m_c}$ (Lemma \ref{lmm:beta lower bound}) and $\beta(\sigma^n)\le1$.
By Bernoulli's inequality,
\begin{align*}
    \beta(\sigma^n)=\max_e\p{1-(1-\sigma^E(e))^n}\le n\max_e\sigma^E(e)=n\beta(\sigma)
\end{align*}
so $\beta(\sigma^n)\le\min(m_c,n)\beta(\sigma)$.
Now we prove that $\beta(\sigma)\le\beta(\sigma^n)$. Suppose that $\beta(\sigma)=\sigma^E(e)$. As $\beta(\sigma^n)\eqdef\max_{e'}\p{1-(1-\sigma^E(e')^n}$, by monotonicity $\beta(\sigma^n)=1-(1-\sigma^E(e))^n\ge\sigma^E(e)=\beta(\sigma)$.
\qed\end{proof}

\begin{proof}[Proof of Lemma \ref{lmm:r upper bound n to 1}]
Take $\nu\in\argmin_{\sigma\in\Delta(P_{st})}F_1(\sigma,\alpha)$. Then
$$r(G,n,k,\alpha)=\frac{\min_{\sigma\in(\Delta(P_{st}))^n}F_n(\sigma,\alpha)}{F_1(\nu,\alpha)}\le\frac{F_n(\nu^n,\alpha)}{F_1(\nu,\alpha)}\le x$$
\qed\end{proof}

\begin{proof}[Proof of Lemma \ref{lmm:convhull}]
Let $P_i=\s{p_1,...,p_i}$. Let $X=(x_1,...,x_m)\in S$. Define $y_m=mx_m$ and $y_i=i(x_i-x_{i+1})$ for every $1\le i<m$. Then
$y_i\in[0,1]$, $\sum y_i=\sum x_i=1$, and treating $X,U_{P_i}$ as vectors in $\R^m$ we get $\sum y_iU_{P_i}=X$. Hence $X\in\convhull\p{\s{U_{P_i}}_{i\in[m]}}\eqdef\convhull\p{\mathcal{U}_1}$. The other inclusion follows by definition.

Now we prove the second equality in the statement. Because all distributions in its right-hand side (rhs) are iid,
\begin{align*}
    \s{X^n\mid X\in\convhull(\mathcal{U}_1)}&=\s{\p{\sum_ix_iU_{P_i}}_{j\in[n]}\mid x_i\in[0,1],~~\sum_ix_i=1}\\
    &=\s{\sum_ix_iU_{P_i}^n\mid x_i\in[0,1],~~\sum_ix_i=1}\\
    &\eqdef\convhull(\mathcal{U}_n)
\end{align*}
where we treated $U_{P_i}^n$ (and $X$) as a vector in $\R^{m\times n}$.
\qed\end{proof}

\section{Connections of Previous Results to Our Model}\label{sec:known and slider}
\paragraph{Converting payoffs to costs} \citet{basilico2017team} studied games of team agents vs. an adversary, where the team agents aim to {\it maximize their payoff}. For a payoff function $u$, let $\Tilde{r}(G,n,k,u)=\frac{\max_{\sigma\in\Delta((P_{st})^n)}\min_\tau u(\sigma,\tau)}{\max_{\sigma\in(\Delta(P_{st}))^n}\min_\tau u(\sigma,\tau)}$. From \cite{von1997team} it follows that $\Tilde{r}$ is indeed the ratio of the cost at TME to the cost at TMECor. \citet{basilico2017team} showed that for a $G$ that contains only $s,t$ and $m$ disjoint paths of length 1 between them, the payoff function $u(a,d)=1\s{\bigcup a\cap\bigcup d=\emptyset}$ and $k=m-1$, the PoU is $\Tilde{r}(G,n,k,u)=m^{n-1}$.
On the other hand, they showed that for every such zero-sum game, this PoU is at most $m^{n-1}$. By allowing arbitrary payoffs, this formulation captures every normal-form game that is considered in~\cite{basilico2017team}. The same results hold when we assume that the team agents minimize a cost function. This is because one can transform a payoff function to a cost function while preserving the PoU (up to an additional assumption):
\begin{lemma}\label{lmm:c to u}
Let $G,n,k$ and a payoff function $u$ such that
$$\max_{\sigma\in(\Delta(P_{st}))^n}\min_{d\in E^k}u(\sigma,d)>0$$
Then there is a cost function $C$ such that $r(G,n,k,C)=\Tilde{r}(G,n,k,u)$.
\end{lemma}
\begin{proof}
Define
$$T'=\max_{\sigma\in\Delta((P_{st})^n)}\min_{d\in E^k}u(\sigma,d),~T''=\max_{\sigma\in(\Delta(P_{st}))^n}\min_{d\in E^k}u(\sigma,d)$$
Let $T=T'+T''$, and define the cost function $C(a,d)=T-u(a,d)$.
By linearity, $C(\sigma,\tau)=\E_{(a,d)\sim(\sigma,\tau)}(T-u(a,d))=T-u(\sigma,\tau)$, so
\begin{align*}
    r(G,n,k,C)&=\frac{\min_{\sigma\in(\Delta(P_{st}))^n}\max_{d\in E^k}(T-u(\sigma,d))}{\min_{\sigma\in\Delta((P_{st})^n)}\max_{d\in E^k}(T-u(\sigma,d))}\\
    &=\frac{T-\max_{\sigma\in(\Delta(P_{st}))^n}\min_{d\in E^k}u(\sigma,d)}{T-\max_{\sigma\in\Delta((P_{st})^n)}\min_{d\in E^k}u(\sigma,d)}
    =\frac{T-T''}{T-T'}=\frac{T'}{T''}=\Tilde{r}(G,n,k,u)
\end{align*}
\qed\end{proof}

When converting the payoff function of \citet{basilico2017team}, we get the cost function $C(a,d)=\frac{1}{m}+\frac{1}{m^n}-1\s{\bigcup a\cap\bigcup d=\emptyset}$. A similar proof shows how to convert a cost function to a payoff function while preserving the PoU, under an analogous assumption on the cost function.

\section{Omitted Proofs in Section \ref{sec:sum}}\label{appendix:sum}
\begin{proof}[Proof of Lemma \ref{lmm:sum cost F form}]
\begin{align*}
    C(\sigma,d,\alpha)&=\E_{a\sim\sigma}\p{\abs{\bigcup a}}+\alpha\sum_{a\in (P_{st})^n}\sigma(a)\sum_{j\le k}1\s{d_j\in\bigcup a}\\
    &=\E_{a\sim\sigma}\p{\abs{\bigcup a}}+\alpha\sum_{j\le k}\sum_{a\in (P_{st})^n}\sigma(a)1\s{d_j\in\bigcup a}\\
    &\eqdef\E_{a\sim\sigma}\p{\abs{\bigcup a}}+\alpha\sum_{j\le k}\sigma^E(d_j)
\end{align*}
Now we see that $\max_dC(\sigma,d,\alpha)$ is obtained for $d$ such that $d_j\in\argmax_e\sigma^E(e)$ for every $j$. For such a $d$, the above becomes
\begin{align*}
    \E_{a\sim\sigma}\p{\abs{\bigcup a}}+\alpha\sum_{j\le k}\beta(\sigma)=\E_{a\sim\sigma}\p{\abs{\bigcup a}}+k\alpha\beta(\sigma)
\end{align*}
Now suppose $P:=\bigcup_j\supp\sigma_j$ is a set of disjoint paths. Then for every $e$ edge, $\sigma^E(e)=\sum_{p\in P_{st,e}}\sigma(p)$ as a sum of disjoint events. So
\begin{align*}
    \E_{a\sim\sigma}\p{\abs{\bigcup a}}&=\sum_{e\in E}\sigma^E(e)=\sum_{e\in E}\sum_{p\in P_{st,e}}\sigma(p)=\sum_{e\in E}\sum_{p\in P_{st}}1_{e\in p}\sigma(p)\\
    &=\sum_p\sum_{e\in p}\sigma(p)=\sum_p\sigma(p)|p|
\end{align*}
\qed\end{proof}

\begin{proof}[Proof of Theorem \ref{thm:general cost uniform is approx}]
Let $f\in MF$ with $|f|=m_c$. Then by definition, $F_n(U_{P_f}^n,\alpha)\le n\abs\pmax+\frac{\emn{m_c}{n}}{m_c}k\alpha$. On the other hand, $F_n(\sigma,\alpha)\ge \abs\pmin+k\alpha\beta(\sigma)$. By Lemma \ref{lmm:beta lower bound} we get the lower bound $\abs\pmin+\frac{\emn{m_c}{n}}{m_c}k\alpha$. Therefore, for the inequality to hold it is sufficient to have
\begin{align*}
    n\abs\pmax+\frac{\emn{m_c}{n}}{m_c}k\alpha\le(1+\e)\p{b+\frac{\emn{m_c}{n}}{m_c}k\alpha}
\end{align*}
and this holds for our choice of $\alpha$.
\qed\end{proof}

\begin{proof}[Proof of Theorem \ref{thm:general costs r bounds}]
We begin with proving Item 1. By Lemma \ref{lmm:r upper bound n to 1}, $r(G,n,k,\alpha)\le n$.
We have $F_1(\sigma,\alpha)=L(\sigma)+k\alpha\beta(\sigma)\ge b+\frac{k\alpha}{m_c}$, 
and
$$\min_\sigma F_n(\sigma,\alpha)\le F_n(1_{\pmin}^n,\alpha)=F_1(1_{\pmin},\alpha)=\abs\pmin+k\alpha$$
So in total
\begin{align*}
    r(G,n,k,\alpha)\le\min\p{\frac{\abs\pmin+k\alpha}{\abs\pmin+\frac{k\alpha}{m_c}},n}\le\min(m_c,n)
\end{align*}

Now we prove Item 2. $F_n(\sigma,\alpha)\le F_n(U_{P_f}^n,\alpha)\le n\abs\pmax+\frac{\emn{m_c}{n}}{m_c}k\alpha$.
In total we have for the chosen $\alpha$ that
\begin{align*}
    r(G,n,k,\alpha)\le\frac{B+\frac{\emn{m_c}{n}}{m_c}k\alpha}{b+\frac{k\alpha}{m
    _c}}\le(1+\e)\emn{m_c}{n}
\end{align*}

Now Item 3. By definition $F_1(\sigma,\alpha)\le F_1(U_{P_f},\alpha)\le n\abs\pmax+\frac{k\alpha}{m_c}$ for $f\in MF$ with $|f|=m_c$, and $F_n(\sigma,\alpha)\ge \abs\pmin+k\alpha\beta(\sigma)\ge\abs\pmin+\frac{\emn{m_c}{n}}{m_c}k\alpha$. In total we have for $\alpha$ as in the statement that
\begin{align*}
    r(G,n,k,\alpha)\ge\frac{\abs\pmin+\frac{\emn{m_c}{n}}{m_c}k\alpha}{n\abs\pmax+\frac{k\alpha}{m
    _c}}\ge(1-\e)\emn{m_c}{n}
\end{align*}

Item 4 follows from Items 2 and 3 because $\e\to0$ as $\alpha\to\infty$, completing the proof.
\qed\end{proof}

\begin{proof}[Proof of Theorem \ref{thm:r_U characterization}]
The idea is that $r_U$ is a ratio of two minimums of sets of affine functions. Define:
\begin{align*}
    \varphi(\alpha,n)=\min_{U^n\in\mathcal{U}_n}F_n(U^n,\alpha)=\min_{i\in[m_c]}\p{\frac{\emn{i}{n}}{i}\p{k\alpha+S_i}}=:\min_{i\in[m_c]}h_i(\alpha,n)
\end{align*}
Note that $h_i$ is an affine function in $\alpha$, and that $r_U(G,n,k,\alpha)=\frac{\varphi(\alpha,n)}{\varphi(\alpha,1)}$. Therefore $r_U$ is smooth a.s. and continuous.
We show that $\frac{\varphi(\alpha,n)}{\varphi(\alpha,1)}$ is strictly increasing in $\alpha$. Fix $\alpha_0$. We ignore the countable set of points where more than one $h_i$ attains $\varphi(\alpha,n)$ for a fixed $\alpha$, because they are all isolated and the PoU is a continuous function in $\alpha$. Suppose $\varphi(\alpha_0,1)=h_i(\alpha_0,1)$ and let $j>i$. Then $h_i(\alpha_0,1)<h_j(\alpha_0,1)$.
For every $n$, the coefficient of $\alpha$ in $h_i(\alpha,n)$ is larger than that in $h_j(\alpha,n)$, so after the second argument of $h_i,h_j$ is fixed to 1, the unique intersection point $\alpha_1$ of the functions $h_i(\alpha,1),h_j(\alpha,1)$ satisfies $\alpha_1>\alpha_0$. Fix the second argument to $n>1$ now, and let $\alpha_n$ be the unique intersection point of $h_i(\alpha,n),h_j(\alpha,n)$. We show that $\alpha_n\ge\alpha_1$, which implies by $h_i,h_j$ being affine in $\alpha$ and the above that $h_i(\alpha_0,n)<h_j(\alpha_0,n)$.

Let $E_i=\emn{i}{n}$ and $E_j=\emn{j}{n}$. Then $\alpha_n$ satisfies
\begin{align*}
    \frac{E_j}{j}(k\alpha_n+S_j)&=\frac{E_i}{i}(k\alpha_n+S_i)\iff\\
    k\alpha_n\p{\frac{E_i}{i}-\frac{E_j}{j}}&=S_j\frac{E_j}{j}-S_j\frac{E_i}{i}\iff\\
    \alpha_n&=\frac{1}{k}\frac{S_j\frac{E_j}{j}-S_j\frac{E_i}{i}}{\frac{E_i}{i}-\frac{E_j}{j}}=\frac{1}{k}\frac{iS_jE_j-jS_iE_i}{jE_i-iE_j}\\
    \Rightarrow~\alpha_1&=\frac{1}{k}\frac{iS_j-jS_i}{j-i}
\end{align*}
and
\begin{align*}
    \alpha_n&\ge\alpha_1\iff\\
    \frac{iS_jE_j-jS_iE_i}{jE_i-iE_j}&\ge\frac{iS_j-jS_i}{j-i}\iff\\
    (j-i)(iS_jE_j-jS_iE_i)&\ge(jE_i-iE_j)(iS_j-jS_i)\iff\\
    ij(S_jE_j+S_iE_i)-i^2S_jE_j-j^2S_iE_i&\ge ij(S_jE_i+S_iE_j)-i^2S_jE_j-j^2S_iE_i\iff\\
    S_jE_j+S_iE_i&\ge S_jE_i+S_iE_j\iff\\
    S_j(E_j-E_i)&\ge S_i(E_j-E_i)
\end{align*}
and this holds because $S_j>S_i$ and $E_j>E_i$. Hence, the coefficient of $\alpha_0$ in $\varphi(\alpha_0,1)$ is $\frac{k}{i}$, and in $\varphi(\alpha_0,n)$ it is at least $\frac{kE_i}{i}$. Since $E_i>1$, in total $\frac{d}{d\alpha}\varphi(\alpha,n)\mid_{\alpha=\alpha_0}>\frac{d}{d\alpha}\varphi(\alpha,1)\mid_{\alpha=\alpha_0}$, so by continuity and smoothness a.s., the ratio $\frac{\varphi(\alpha,n)}{\varphi(\alpha,1)}$ is strictly increasing for every $\alpha>0$.

Therefore for sufficiently large $\alpha$, $\varphi(\alpha,n)$ equals the affine function with the smallest leading coefficient: there exists $A(G,n,k)$ such that
\begin{align*}
    \forall\alpha\ge A(G,n,k)~\forall n'\le n~~~\varphi(\alpha,n')=\frac{\emn{m_c}{n'}}{m_c}\p{k\alpha+S_{m_c}}
\end{align*}
so
\begin{align*}
    \forall\alpha\ge A(G,n,k)~~~r_U(G,n,k,\alpha)=\frac{\varphi(\alpha,n)}{\varphi(\alpha,1)}=\emn{m_c}{n}
\end{align*}
Now let $\alpha<A(G,n,k)$. Observe that for some $i\in[m_c]$
\begin{align*}
    \emn{m_c}{n}\varphi(\alpha,1)&=\emn{m_c}{n}\frac{\emn{i}{1}}{i}\p{k\alpha+S_i}=\frac{\emn{m_c}{n}}{i}\p{k\alpha+S_i}\\
    &\stackrel{i,m_c\ge1}{\ge}\frac{\emn{i}{n}}{i}\p{k\alpha+S_i}\ge\varphi(\alpha,n)
\end{align*}
therefore $r_U(G,n,k,\alpha)\le\emn{m_c}{n}$, completing the proof.
\qed\end{proof}

\begin{proof}[Proof of Corollary \ref{coro:computing r_U}]
Note that finding each $P^*_i$ requires solving a mincost flow problem: the flow network is $G$, the edge costs are all 1, and the demands are $i$ for $s$, $-i$ for $t$ and 0 for all other $v\in V$. Therefore algorithms such as those of \citet{chen2022maximum} and \citet{orlin1988faster} yield the above running times.
\qed\end{proof}

\begin{proof}[Proof of Theorem \ref{thm:alpha bound}]
We reuse the notations from Theorem \ref{thm:r_U characterization}. By the proof of the latter, $\alpha\ge A(G,n,k)$ when $\varphi(\alpha,n)=h_{m_c}(\alpha,n)$.
This holds iff for every $i<m_c$,
\begin{align*}
    \alpha&\ge\frac{1}{k}\frac{iS_{m_c}E_{m_c}-m_cS_iE_i}{m_cE_i-iE_{m_c}}=\frac{1}{k}\frac{S_{m_c}\p{1-\p{1-\frac{1}{m_c}}^n}-S_i\p{1-\p{1-\frac{1}{i}}^n}}{\p{1-\frac{1}{m_c}}^n-\p{1-\frac{1}{i}}^n}
\end{align*}
Rhs is lower-bounded by $\frac{1}{k}\frac{S_{m_c}\p{1-\p{1-\frac{1}{m_c}}^n}-S_{m_c-1}}{\p{1-\frac{1}{m_c}}^n}$.
Using the well-known equality $\p{1-\frac{1}{m_c}}^n=e^{-\Theta(n/m_c)}$ we get
\begin{align*}
    &=\frac{1}{k}\frac{S_{m_c}\p{1-e^{-\Theta(n/m_c)}}-S_{m_c-1}}{e^{-\Theta(n/m_c)}}\\
    &=\frac{1}{k}\p{S_{m_c}\p{e^{\Theta(n/m_c)}-1}-S_{m_c-1}e^{\Theta(n/m_c)}}\\
    &=\frac{e^{\Theta(n/m_c)}\p{S_{m_c}-S_{m_c-1}}-S_{m_c}}{k}
\end{align*}
\qed\end{proof}

\begin{proof}[Proof of Lemma \ref{lmm:assume decreasing sum}]
Let $\sigma$ and $|p_i|<|p_{i'}|$. For the distribution $\sigma'$ that results from $\sigma$ by swapping the weights on $p_i,p_i$ we have $F_n(\sigma,\alpha)-F_n(\sigma',\alpha)=\sigma(p_i)(|p_i|-|p_{i'}|)+\sigma(p_{i'})(|p_{i'}|-|p_i|)=(|p_{i'}|-|p_i|)(\sigma(p_{i'})-\sigma(p_i))$ by Lemma \ref{lmm:sum cost F form}. This is $\le0$ iff $\sigma(p_i)\ge\sigma(p_{i'})$. When $|p_i|=|p_{i'}|$, the difference above is 0 so we reorder the weights of $\sigma$ on $p_i,p_{i'}$ to make the property in the lemma statement hold.
\qed\end{proof}

\begin{proof}[Proof of Lemma \ref{lmm:assume iid sum}]
Let $\sigma$. Take the distribution $X=\p{\frac{\sum_{j\in[n]}\sigma_{ij}}{n}}_{i\in[m]}$, and $\mu=X^n$. Then for every path $p_i$,
\begin{align*}
    \sigma(p_i)\eqdef1-\prod_j(1-\sigma_{ij})&\ge1-\p{1-\frac{\sum_{j\in[n]}\sigma_{ij}}{n}}^n=\mu(p_i)\\
    \iff\prod_j(1-\sigma_{ij})&\le\p{1-\frac{\sum_{j\in[n]}\sigma_{ij}}{n}}^n\\
    \iff\sqrt[n]{\prod_j(1-\sigma_{ij})}&\le1-\frac{\sum_{j\in[n]}\sigma_{ij}}{n}
\end{align*}
The last inequality follows from the AMGM inequality (with equality iff $\{\sigma_{ij}\}_j$ are all equal). Therefore by monotonicity with respect to the probabilities $\sigma(p_i)$, $F_n(\sigma,\alpha)\ge F_n(\mu,\alpha)$.
\qed\end{proof}

\begin{proof}[Proof of Theorem \ref{thm:min F disjoint paths}]
For every $X\in\convhull(\mathcal{U}_1)$, by concavity of $F^1_n$ there holds $F^1_n(X)\ge\min_{X\in\mathcal{U}_1}F^1_n\p{X}$. Then, by the lemmas above,
$$\min_{X\in\mathcal{U}_1}F^1_n\p{X}=\min_{\sigma\in\mathcal{U}_n}F_n(\sigma,\alpha)=\min_{\sigma\in(\Delta(P_{st}))^n}F_n(\sigma,\alpha)$$
$r=r_U$ follows by definition.
\qed\end{proof}

\section{Omitted Details in Section \ref{sec:max}}\label{sec:max details}

\subsection{Illustrating example}\label{subapp:max eg}
\begin{example}\label{ex:max cost not min at u}
Consider the graph in Figure \ref{fig2}:
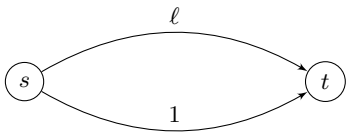
\begin{figure}[H]
\centering
\begin{tikzpicture}
    \tikzset{vertex/.style = {shape=circle,draw,minimum size=1em}}
    \tikzset{edge/.style = {->,> = latex'}}
    \node[vertex] (s) at (0,0) {$s$};
    \node[vertex] (t) at (4,0) {$t$};
    \draw[edge] [bend left] (s) to node[above]{$\ell$} (t);
    \draw[edge] [bend right] (s) to node[above]{$1$} (t);
\end{tikzpicture}
\caption{Illustrating example of Theorem \ref{thm:max uniform is almost never optimal}}
\label{fig2}
\end{figure}
where $1<\ell\in\N$ are such that the 2 edges represent paths of the corresponding lengths. Let $n>k=1$ and $\alpha\ge\ell-1$. Lemma \ref{lmm:assume iid max} will show that we can assume for minimization that the agents' strategies are iid. Then we can identify their iid strategies with $x\in[0,1]$, where $x$ is the probability that a agent chooses the short (bottom) path. Similarly, because the paths are disjoint, we identify the interceptor's strategy with $y\in[0,1]$. Now we can write in short $C(x,y,\alpha)$ for the corresponding cost. We show that for every such $\alpha$, $x^*:=\argmin_x\max_yC(x,y,\alpha)\in\left(2^{-1},2^{-1/n}\right]$, such that $x^*=2^{-1/n}$ for $\alpha=\ell-1$, and $x^*$ strictly decreases as $\alpha$ increases. In particular, $\min_{\sigma\in(\Delta(P_{st}))^n}\max_\tau C(\sigma,\tau,\alpha)=:C(x^*,y^*,\alpha)<\min_{\sigma\in\mathcal{U}_n}\max_\tau C(\sigma,\tau,\alpha)$.

For $x\in[0,1]$, we have
\begin{align*}
    f_1(x,\alpha)=C(x,1,\alpha)&=(1-x)^n\ell+(1-(1-x)^n)(1+\alpha)\\
    f_2(x,\alpha)=C(x,0,\alpha)&=x^n+(1-x^n)(\ell+\alpha)
\end{align*}
$f_1$ is increasing in $x$, and $f_2$ is decreasing in $x$, so they intersect at at most one $x$. Also, $\max_yC(x,y,\alpha)=\max(f_1(x,\alpha),f_2(x,\alpha))$. So $\min_x\max(f_1(x,\alpha),f_2(x,\alpha))$ is obtained when $f_1(x,\alpha)=f_2(x,\alpha)$. At $x=\frac{1}{2}$ we have
\begin{align*}
    f_1\p{\frac{1}{2},\alpha}=2^{-n}\ell+(1-2^{-n})(1+\alpha)<2^{-n}+(1-2^{-n})(\ell+\alpha)=f_2\p{\frac{1}{2},\alpha}
\end{align*}
iff $(\ell-1)(1-2^{-n+1})>0$, which is the case here. Hence $f_1,f_2$ intersect after $x=\frac{1}{2}$. At $x=1$ we get
\begin{align*}
    f_1(1,\alpha)=1+\alpha>1=f_2(1,\alpha)
\end{align*}
so $f_1,f_2$ intersect in $\p{\frac{1}{2},1}$, therefore $\min_x\max(f_1(x,\alpha),f_2(x,\alpha))\in\p{\frac{1}{2},1}$. The exact range $\left(2^{-1},2^{-1/n}\right]$ is obtained from solving $f_1(x,\alpha)=f_2(x,\alpha)$ for $\alpha\ge\ell-1$. This finishes the proof.
\end{example}

\subsection{More definitions}
We need more definitions to prove the results of Section \ref{sec:max}. For $\alpha$ and $d\in E^k$, let $c_{d,\alpha}:E\to\R$ by defined by $c_{d,\alpha}(e)=1+\alpha\abs{\s{i\in[k]\mid d_i=e}}$ i.e., the cost to unlock edge $e$ given $\alpha,d$.
When $\alpha$ is clear from context, we will write $c_d$. We abuse notation and for $p\in P_{st}$ write $c_d(p)=\sum_{e\in p}c_d(p)$, and $c_d(a)=(c_d(a_1),...,c_d(a_n))$ for $a\in (P_{st})^n$. Let $c_0=c_{d,0}$ denote the default cost function, where there are no adversaries. Let $P_{st}^d=\{p^d_i\}_{i\in[m]}$ denote the paths $P_{st}$ where the paths indices are increasing with respect to $c_d$. Let $I_d:[m]\to[m]$ be the corresponding mapping of indices from $\{p^d_i\}_{i\in[m]}$ to $\{p_i\}_{i\in[m]}$.
Let $k^d_i$ be the number of interceptors who picked edges on path $p^d_i$, and define $k^d_0=0$. By definition $c_d(p^d_i)=|p_{I_d(i)}|+k^d_i\alpha$.\footnote{Here we won't consider powers of $k$ as an integer.} If $k=1$ we simplify notation by writing $c_i,k^i_{i'},I_i,P_{st}^i=\{p^i_{i'}\}_{i'\in[m]}$.

\subsection{Expressions for MAX cost}

Here we derive some expressions for $f_d(\sigma,\alpha)$ that will help in future proofs. For $d\in E^k$, let $c_d,k^d_i,p^d_i$ as defined in Section \ref{sec:max}. Using the equality $\E(X)=\sum_x\Pr(X\ge x)$ for a random variable $X$, we derive
\begin{align}\label{eq:fd max cost}
    f_d(\sigma,\alpha)=&\sum_{i\in[m]}(c_d(p^d_i)-c_d(p^d_{i-1}))\Pr_{a\sim\sigma}\p{\norm{c_d(a)}_\infty\ge c_d(p^d_i)}\\
\notag
    =&\sum_{i\in[m]}(c_d(p^d_i)-c_d(p^d_{i-1}))\p{1-\prod_{j\in[n]}\p{\sum_{i'\le i-1}\sigma_j(p^d_{i'})}}\\
\notag
    =&\sum_{i\in[m]}c_d(p^d_i)\p{\Pr_{a\sim\sigma}\p{\norm{c_d(a)}_\infty\ge c_d(p^d_i)}-\Pr_{a\sim\sigma}\p{\norm{c_d(a)}_\infty>c_d(p^d_i)}}\\
\notag
    =&\sum_{i\in[m]}c_d(p^d_i)\p{\prod_{j\in[n]}\p{\sum_{i'\le i}\sigma_j(p^d_{i'})}-\prod_{j\in[n]}\p{\sum_{i'\le i-1}\sigma_j(p^d_{i'})}}
\end{align}

We write Equation \ref{eq:fd max cost} more explicitly when $k=1$. For every $i<m$:
\begin{align}\label{eq:fi max cost disjoint i<m}
    f_i(\sigma,\alpha)=&\sum_{i'\in[m]}(c_i(p^i_{i'})-c_i(p^i_{i'-1}))\Pr_{a\sim\sigma}\p{\norm{c_i(a)}_\infty\ge c_i(p^i_{i'})}\\
\notag
    =&\sum_{i'<i}(|p_{i'}|-|p_{i'-1}|)\p{1-\prod_{j\in[n]}\p{\sum_{i''<i'}\sigma_j(p_{i''})}}~+\\
\notag
    &(|p_{i+1}|-|p_{i-1}|)\p{1-\prod_{j\in[n]}\p{\sum_{i''<i}\sigma_j(p_{i''})}}~+\\
\notag
    &\sum_{i'>i+1}(|p_{i'}|-|p_{i'-1}|)\p{1-\prod_{j\in[n]}\p{\p{\sum_{i''<i'}\sigma_j(p_{i''})}-\sigma_j(p_i)}}~+\\
\notag
    &(|p_i|+\alpha-|p_m|)\p{1-\prod_{j\in[n]}\p{1-\sigma_j(p_i)}}
\end{align}
For $i=m$:
\begin{align}\label{eq:fi max cost disjoint i=m}
    f_m(\sigma,\alpha)=&\sum_{i'\in[m]}(c_m(p^m_{i'})-c_m(p^m_{i'-1}))\Pr_{a\sim\sigma}\p{\norm{c_m(a)}_\infty\ge c_m(p^m_{i'})}\\
\notag
    =&\sum_{i'<m}(|p_{i'}|-|p_{i'-1}|)\p{1-\prod_{j\in[n]}\p{\sum_{i''<i'}\sigma_j(p_{i''})}}~+\\
\notag
    &(|p_m|+\alpha-|p_{m-1}|)\p{1-\prod_{j\in[n]}\p{1-\sigma_j(p_m)}}
\end{align}
So $F_n(\sigma,\alpha)=\max_{i\in[m]}f_i(\sigma,\alpha)$.

\subsection{Omitted proofs}\label{subapp:max proofs}
\begin{proof}[Proof of Lemma \ref{lmm:max Fn bounds}]
For the lower bound, take $\sigma\in(\Delta(P_{st}))^n$. The interceptors can respond by having all pick an edge $e\in\argmax_{e'}\sigma^E(e')$ (the same $e$). 
Since $k\alpha>\abs\pmax-\abs\pmin$, with probability $\beta(\sigma)$ over $a\sim\sigma$ the cost $C(a,e^k,\alpha)$ will have the form $|p|+k\alpha$ for $p\ni e$. So
the cost is $C(\sigma,e^k,\alpha)\ge\E_{a\sim\sigma}(\norm{c_d(a)}_\infty-1_{e\in a}\cdot k\alpha)+k\alpha\beta(\sigma)$. Note that $\norm{c_d(a)}_\infty-1_{e\in a}\cdot k\alpha\in[\abs\pmin,\abs\pmax]$ for every $a$, because this is the length of the longest path in $a$ according to $c_d$, with subtracting the additional $k\alpha$ term, if $e$ is contained in that path.
Then, the cost to the agents is $C(\sigma,e^k,\alpha)\ge\abs\pmin+k\alpha\beta(\sigma)$. Therefore $F_n(\sigma,\alpha)\ge C(\sigma,e^k,\alpha)\ge\abs\pmin+k\alpha\beta(\sigma)$.

For the upper bound: as MAX can be written as a norm, we can use the triangle inequality:
\begin{align*}
    C(a,d,\alpha)&=\norm{c_d(a)}_\infty\le\norm{c_0(a)}_\infty+\norm{c_d(a)-c_0(a)}_\infty\\
    &\le\norm{c_0(a)}_\infty+\alpha\sum_{j\le k}1\s{d_j\in\bigcup a}=(*)
\end{align*}
The norm in $(*)$ is bounded by $\abs\pmax$. In the proof of Lemma \ref{lmm:sum cost F form} we show that the sum in $(*)$ is upper bounded by $k\beta(\sigma)$. Therefore,
$(*)\le\abs\pmax+k\alpha\beta(\sigma)$.
So $F_n(\sigma,\alpha)\le\abs\pmax+k\alpha\beta(\sigma)$.
\qed\end{proof}

\begin{proof}[Proof of Lemma \ref{lmm:assume iid max}]
Let $X(p^d_i)=\p{\frac{\sum_{j\in[n]}\sigma_{j}(p^d_i)}{n}}_{i\in[m]}$, and $\mu=X^n$. Then for every $i$, by the AMGM inequality:
\begin{align*}
    \prod_{j\in[n]}\p{\sum_{i'\le i-1}\sigma_j(p^d_{i'})}\le\p{\frac{1}{n}\sum_{j\in[n]}\sum_{i'\le i-1}\sigma_j(p^d_{i'})}^n=\p{\sum_{i'\le i-1}X(p^d_{i'})}^n
\end{align*}
Therefore by Equation \ref{eq:fd max cost}, for every $d\in E^k$:
\begin{align*}
    f_d(\sigma,\alpha)\ge\sum_{i\in[m]}(c_d(p^d_i)-c_d(p^d_{i-1}))\p{1-\p{\sum_{i'\le i-1}X(p^d_{i'})}^n}\eqdef f_d(\mu,\alpha)
\end{align*}
So $F_n(\sigma,\alpha)\ge F_n(\mu,\alpha)$, proving the lemma.
\qed\end{proof}

\begin{proof}[Proof of Lemma \ref{lmm:max Fn(U)=f_m^k(U)}]
Let $\sigma=U_{P_{st}}^n$, and let $d\in E^k$. Consider the resulting cost $c_d$
and the paths $P_{st}^d$.
Applying Equation \ref{eq:fd max cost}:
\begin{align}\label{eq:lmm 6.4 proof}
\notag
    f_d\p{U_{P_{st}}^n,\alpha}&=\sum_{i=1}^m(c_d(p^d_i)-c_d(p^d_{i-1}))\p{1-\p{\frac{i-1}{m}}^n}\\
\notag
    &=\sum_{i=1}^m(|p_{I_d(i)}|-|p_{I_d(i-1)}|+(k^d_i-k^d_{i-1})\alpha)\p{1-\p{\frac{i-1}{m}}^n}\\
    &=\alpha\sum_{i=1}^mk^d_i\frac{i^n-(i-1)^n}{m^n}+\sum_{i=1}^m|p_{I_d(i)}|\frac{i^n-(i-1)^n}{m^n}
\end{align}
We prove the lemma by showing that (\ref{eq:lmm 6.4 proof}) attains its strict maximum over $d\in E^k$ when $d=(m,...,m)$. (\ref{eq:lmm 6.4 proof}) is an affine function in $\alpha$. $m,n>1$ so $\frac{i^n-(i-1)^n}{m^n}$ is strictly increasing in $i$. Therefore
\begin{align*}
    \alpha\sum_{i=1}^mk^d_i\frac{i^n-(i-1)^n}{m^n}\le\alpha k\frac{m^n-(m-1)^n}{m^n}=\frac{\emn{m}{n}}{m}k\alpha
\end{align*}
This maximum with respect to $d\in E^k$ is obtained only if $k^d_i=k\cdot1_{i=m}$.
For $d=(m,...,m)$, we indeed get $k^d_i=k\cdot1_{i=m}$: because $\alpha>\frac{|p_m|-|p_1|}{k}$, so for that $d$ the unique longest path will be $p^d_m$ with length $|p_m|+k\alpha$. On the other hand, if $d\ne(m,...,m)$, then $k^d_m<k$ because the interceptors did not all pick edges from the same path. Therefore, the coefficient of $\alpha$ in (\ref{eq:lmm 6.4 proof}) is strictly maximal iff $d=(m,...,m)$. As for $\sum_{i=1}^m|p_{I_d(i)}|\frac{i^n-(i-1)^n}{m^n}$, because $|p_1|\le...\le|p_m|$, we argue that this term is maximal when $I_d(i)=i$, up to reordering of paths of identical lengths: Suppose that $I_d(i)\ne i$ for some $i$. Then $|p_{I_d(j)}|>|p_{I_d(j+1)}|$ for some $j$ (after reordering paths of identical lengths). Then, defining $I'$ by swapping $I_d(j),I_d(j+1)$ yields
\begin{align*}
    &\sum_{i=1}^m|p_{I'(i)}|\frac{i^n-(i-1)^n}{m^n}-\sum_{i=1}^m|p_{I_d(i)}|\frac{i^n-(i-1)^n}{m^n}\\
    =&(|p_{I'(j)}|-|p_{I_d(j)}|)\frac{j^n-(j-1)^n}{m^n}+(|p_{I'(j+1)}|-|p_{I_d(j+1)}|)\frac{(j+1)^n-j^n}{m^n}\\
    =&(|p_{I_d(j+1)}|-|p_{I_d(j)}|)\frac{j^n-(j-1)^n}{m^n}+(|p_{I_d(j)}|-|p_{I_d(j+1)}|)\frac{(j+1)^n-j^n}{m^n}\\
    =&(|p_{I_d(j)}|-|p_{I_d(j+1)}|)\p{\frac{(j+1)^n-j^n}{m^n}-\frac{j^n-(j-1)^n}{m^n}}>0
\end{align*}
This proves the maximality of $I_d$ that satisfies $I_d(i)=i$ for all $i$.

So, since $I_{(m,...,m)}(i)\eqdef i$ for all $i$, we have that (\ref{eq:lmm 6.4 proof}), i.e., $f_d\p{U_{P_{st}}^n,\alpha}$, is strictly maximal for $d=(m,...,m)$, i.e., $F_n\p{U_{P_{st}}^n,\alpha}\eqdef f_{(m,...,m)}\p{U_{P_{st}}^n,\alpha}>\max_{d\ne(m,...,m)}f_d\p{U_{P_{st}}^n,\alpha}$.
\qed\end{proof}

\begin{proof}[Proof of Lemma \ref{lmm:max uniforms in equal lengths}]
Consider $G_i$ for $i\le m_1$. Then to apply Lemma \ref{lmm:max Fn(U)=f_m^k(U)} to it the $\alpha$ inequality constraint becomes $\alpha>0$ which holds by definition of $\alpha$. So
\begin{align*}
    F_n(U_{P_i}^n,\alpha)&=f_{(i,...,i)}(U_{P_i}^n,\alpha)\\
    &=\alpha k\sum_{i'=1}^i1_{i'=i}\frac{(i')^n-(i'-1)^n}{i^n}+
    \sum_{i'=1}^i|p_{i'}|\frac{(i')^n-(i'-1)^n}{i^n}\\
    &=|p_1|+\frac{\emn{i}{n}}{i}k\alpha
\end{align*}
The first equality is due to Lemma \ref{lmm:max Fn(U)=f_m^k(U)}. The second is by Equation \ref{eq:fd max cost}, and the third is because the paths of $G_{m_1}$ have identical lengths. The second part of the lemma follows by that $\frac{\emn{i}{n}}{i}$ is decreasing in $i$.
\qed\end{proof}

\begin{proof}[Proof of Lemma \ref{lmm:bound for that}]
Consider $\sigma=U_{P_{m_1+1}}^n$ and $G_{m_1+1}$.
We show that the bound on $\alpha$ in the lemma statement implies that $\alpha>\frac{|p_{m_1+1}|-|p_1|}{k}$, namely that $\frac{E'_{m_1+1}}{k(E'_{m_1}-E'_{m_1+1})}\ge\frac{1}{k}$.
Canceling out $\frac{1}{k}$ from both sides:
\begin{align*}
    \frac{E'_{m_1+1}}{E'_{m_1}-E'_{m_1+1}}\ge1&\iff2E'_{m_1+1}\ge E'_{m_1}\\
    &\iff2\p{1-\frac{1}{m_1+1}}^n-\p{1-\frac{1}{m_1}}^n\le1
\end{align*}
When $n=2$, the latter inequality simplifies to $2\p{1-\frac{1}{m_1+1}}^2-\p{1-\frac{1}{m_1}}^2\le\frac12<1$. We show that $2\p{1-\frac{1}{m_1+1}}^n-\p{1-\frac{1}{m_1}}^n$ decreases as $n$ increases, to conclude that it is at most 1 for every $n>1$.
\begin{align*}
    &2\p{1-\frac{1}{m_1+1}}^{n+1}-\p{1-\frac{1}{m_1}}^{n+1}-2\p{1-\frac{1}{m_1+1}}^n+\p{1-\frac{1}{m_1}}^n\\
    =&\frac{1}{m_1}\p{1-\frac{1}{m_1}}^n-\frac{2}{m_1+1}\p{1-\frac{1}{m_1+1}}^n\le0\\
    \iff~&\frac{m_1+1}{m_1}\frac{\p{1-\frac{1}{m_1}}^n}{\p{1-\frac{1}{m_1+1}}^n}\le2\\
    \iff~&\frac{m_1+1}{m_1}\p{\frac{m_1^2-1}{m_1^2}}^n\le2
\end{align*}
The latter holds because $\frac{m_1+1}{m_1}\le2$ and $\p{\frac{m_1^2-1}{m_1^2}}^n\le1$.

Therefore, we can use Lemma \ref{lmm:max Fn(U)=f_m^k(U)} to obtain
\begin{align*}
    F_n(U_{P_{m_1+1}}^n,\alpha)&=\alpha\sum_{i=1}^{m_1+1}k\cdot1_{i=m_1+1}\frac{i^n-(i-1)^n}{(m_1+1)^n}+\sum_{i=1}^{m_1+1}|p_i|\frac{i^n-(i-1)^n}{(m_1+1)^n}\\
    &=|p_1|+(|p_{m_1+1}|-|p_1|)E'_{m_1+1}+E'_{m_1+1}k\alpha\\
    &<|p_1|+E'_{m_1}k\alpha
\end{align*}
The second equality uses $\frac{(m_1+1)^n-m_1^n}{(m_1+1)^n}=E'_{m_1+1}$ and $\sum_{i=1}^{m_1+1}\frac{i^n-(i-1)^n}{(m_1+1)^n}=1$, and the inequality uses the assumption on $\alpha$. The lemma now follows by definition.
\qed\end{proof}

\begin{proof}[Proof of Theorem \ref{thm:max uniform is almost never optimal}]
By Lemma \ref{lmm:max Fn(U)=f_m^k(U)}, 
\begin{align*}
    F_n(U_{P_{st}}^n,\alpha)=f_{(m,...,m)}(U_{P_{st}}^n,\alpha)>\max_{d\ne(m,...,m)}f_d(U_{P_{st}}^n,\alpha)
\end{align*}
For a small $\e>0$, let $\nu$ such that $\nu(p_1)=\frac{1}{m}+\e,\nu(p_2)=...=\nu(p_{m-1})=\frac{1}{m},\nu(p_m)=\frac{1}{m}-\e$. Observe that for every $d$, $f_d$ is continuous, and $|E^k|<\infty$. Hence, when $\e$ is sufficiently small, we still have $F_n(\nu^n,\alpha)= f_{(m,...,m)}\p{\nu^n,\alpha}>\max_{d\ne(m,...,m)}f_d\p{\nu^n,\alpha}$.

$f_{(m,...,m)}$ is the cost to the agents when the interceptors pick $k$ edges of $p_m$. $\nu$ differs from $U_{P_{st}}$ by shifting $\e$ weight from picking $p_m$ to picking $p_1$, and $|p_1|<|p_m|$ by our assumption. So, when we go from $U_{P_{st}}^n$ to $\nu^n$, we get $f_{(m,...,m)}(U_{P_{st}}^n,\alpha)>f_{(m,...,m)}(\nu^n,\alpha)$, hence $F_n(U_{P_{st}}^n,\alpha)>F_n(\nu^n,\alpha)$, namely $U_{P_{st}}^n$ is not optimal for the agents.

Let $m_1<i<m$, and consider $G_i,U_{P_i}^n$.
$|p_i|\le|p_m|$, so $\alpha>\frac{|p_i|-|p_1|}{k}$. Now there are 2 cases.
In the first case, $\min_{\sigma\in(\Delta(P_i))^n}F_n(\sigma,\alpha)=F_n(U_{P_{m_1}}^n,\alpha)=|p_1|+\frac{\emn{m_1}{n}}{m_1}k\alpha$. Then, in particular
\begin{align*}
    \min_{\sigma\in\mathcal{U}_n\cap(\Delta(P_i))^n}F_n(\sigma,\alpha)\ge|p_1|+\frac{\emn{m_1}{n}}{m_1}k\alpha>\min_{\sigma\in(\Delta(P_{st}))^n}F_n(\sigma,\alpha)
\end{align*}
by the theorem statement.
In the second case, $\min_{\sigma\in(\Delta(P_i))^n}F_n(\sigma,\alpha)<|p_1|+\frac{\emn{m_1}{n}}{m_1}k\alpha$ holds. Then, we apply the argument above to the graph $G_i$ to get $F_n(U_{P_i}^n,\alpha)>\min_{\sigma\in(\Delta(P_{st}))^n}F_n(\sigma,\alpha)$.
In both cases, we conclude that $U_{P_i}^n$ is not optimal for the agents.

Now consider $G_{m_1}$. Then by Lemmas \ref{lmm:max uniforms in equal lengths} and \ref{lmm:bound for that}, we again get that $U_{P_i}^n$ is not optimal for the agents, now for all $i\le m_1$.

In total, $\min_{\sigma\in(\Delta(P_{st}))^n}F_n(\sigma,\alpha)<\min_{\sigma\in\mathcal{U}_n}F_n(\sigma,\alpha)$ by the definition of $\mathcal{U}_n$.
\qed\end{proof}

\begin{proof}[Proof of Theorem \ref{thm:max uniform equalengths is optimal}]
Using Lemma \ref{lmm:assume iid max}, we can wlog consider only $\sigma=X^n$. For $d\in E^k$, by Equation \ref{eq:fd max cost}:
\begin{align*}
    f_d(\sigma,\alpha)&=\sum_{i\in[m]}c_d(p^d_i)\p{\p{\sum_{i'\le i}X(p^d_{i'})}^n-\p{\sum_{i'\le i-1}X(p^d_{i'})}^n}=(*)
\end{align*}
Let $y_i=\sum_{i'\le i}X(p^d_{i'})$, let $\delta(X)=\max_{i\in[m]}(y_i^n-y_{i-1}^n)$ and let $i^*\in\argmax_i(y_i^n-y_{i-1}^n)$. Then
\begin{align*}
    (*)&=\alpha\sum_{i\in[m]}k^d_i\p{y_i^n-y_{i-1}^n}+\sum_{i\in[m]}|p_{I_d(i)}|\p{y_i^n-y_{i-1}^n}\\
    &=\alpha\sum_{i\in[m]}k^d_i\p{y_i^n-y_{i-1}^n}+\sum_{i\in[m]}|p_1|\p{y_i^n-y_{i-1}^n}
\end{align*}
as $|p_1|=...=|p_m|$. So the above is maximal when $k^d_i=k\cdot1_{i=i^*}$, and the cost to the agents is
\begin{align*}
    F_n(\sigma,\alpha)=f_{(i^*,...,i^*)}(\sigma,\alpha)=|p_1|+\delta(X)k\alpha
\end{align*}
If $X=U_{P_{st}}$, then $\delta(X)=\frac{\emn{m}{n}}{m}$, $i^*=m$ and we get
$$F_n(U_{P_{st}}^n,\alpha)=f_{(m,...,m)}(U_{P_{st}}^n,\alpha)=|p_1|+\frac{\emn{m}{n}}{m}k\alpha$$
Now suppose otherwise. The paths have identical lengths, so wlog we reorder them such that $X(p^d_m)>\frac{1}{m}$. Then $\delta(X)\ge y_m^n-y_{m-1}^n\eqdef1-(1-X(p^d_m))^n>\frac{\emn{m}{n}}{m}$. Therefore $F_n(\sigma,\alpha)>|p_1|+\frac{\emn{m}{n}}{m}k\alpha$. Hence $\min_\sigma F_n(\sigma,\alpha)=F_n(U_{P_{st}}^n,\alpha)=|p_1|+\frac{\emn{m}{n}}{m}k\alpha$, completing the proof.
\qed\end{proof}

\begin{proof}[Proof of Lemma \ref{lmm:assume decreasing max}]
Let $\sigma$, and wlog it's iid, namely $\sigma=X^n$ for $X\in\Delta(P_{st})$.
Let $i<m$ such that $|p_i|<|p_{i+1}|$ and $\sigma(p_i)<\sigma(p_{i+1})$.
Let $X'$ result from $X$ by swapping $X_i,X_{i+1}$, and let $\sigma'=(X')^n$. Then $\sigma'$ also results from $\sigma$ by swapping $\sigma(p_i),\sigma(p_{i+1})$.
We split into two cases.

\paragraph{Case $i<m-1$:} Let $i'\in[m]$. First, suppose $i'\ne i$. If $i'\ne i+1$ too then $f_{i'}(\sigma')<f_{i'}(\sigma)$. This is because $\sigma'$ results from $\sigma$ by swapping the probabilities of $p_i,p_{i+1}$ being picked, and $\sigma(p_i)<\sigma(p_{i+1})$, and $c_{i'}(p_i)=|p_i|<|p_{i+1}|=c_{i'}(p_{i+1})$. If $i'=i+1$ then this inequality holds as well, because $c_{i+1}(p_i)=|p_i|<|p_{i+1}|+\alpha=c_{i+1}(p_{i+1})$, so the argument above holds here too. Therefore $i'=i$ is left. We show that $f_i(\sigma')\le f_{i+1}(\sigma)$. Using Equation \ref{eq:fi max cost disjoint i<m}:
\begin{align*}
    &f_{i+1}(\sigma)-f_i(\sigma')=(|p_i|-|p_{i-1}|)\p{1-\p{\sum_{i'<i}X(p_{i'})}^n}-(|p_{i+1}|-|p_{i-1}|)\p{1-\p{\sum_{i'<i}X(p_{i'})}^n}~+\\
    &(|p_{i+2}|-|p_i|)\p{1-\p{\sum_{i'<i+1}X(p_{i'})}^n}-(|p_{i+2}|-|p_{i+1}|)\p{1-\p{\sum_{i'<i+1}X(p_{i'})}^n}~+\\
    &(|p_{i+1}|+\alpha-|p_m|)\p{1-\p{1-X(p_{i+1})}^n}-(|p_i|+\alpha-|p_m|)\p{1-\p{1-X(p_{i+1})}^n}
\end{align*}
Let $x=\sum_{i'<i}X(p_{i'}),y=X(p_i),z=X(p_{i+1})$. Then we get
\begin{align*}
    f_{i+1}(\sigma)-f_i(\sigma')=&(|p_i|-|p_{i-1}|)\p{1-x^n}-(|p_{i+1}|-|p_{i-1}|)\p{1-x^n}~+\\
    &(|p_{i+2}|-|p_i|)\p{1-\p{x+y}^n}-(|p_{i+2}|-|p_{i+1}|)\p{1-\p{x+y}^n}~+\\
    &(|p_{i+1}|+\alpha-|p_m|)\p{1-\p{1-z}^n}-(|p_i|+\alpha-|p_m|)\p{1-\p{1-z}^n}\\
    =&(|p_{i+1}|-|p_i|)\p{1-(1-z)^n-(x+y)^n+x^n}
\end{align*}
$|p_{i+1}|>|p_i|$. Also, there holds $y<z$ and $x+y+z\le1$, so $x+y\le1$ and $x\le1-z$. Then
\begin{align*}
    1-(1-z)^n=z\sum_{i=0}^n(1-z)^i\ge y\sum_{i=0}^nx^i(x+y)^{n-i}=(x+y)^n-x^n
\end{align*}
Hence $f_{i+1}(\sigma)-f_i(\sigma')\ge0$. Therefore by definition $F_n(\sigma',\alpha)\le F_n(\sigma,\alpha)$.

\paragraph{Case $i=m-1$:} Note that the same proof as in the other case proves that $f_{i'}(\sigma')<f_{i'}(\sigma)$ for $i'\ne m-1$. Case $i'=m-1$ is left. In the same way as in the previous case, using Equations \ref{eq:fi max cost disjoint i<m} and \ref{eq:fi max cost disjoint i=m}:
\begin{align*}
    f_m(\sigma)&-f_{m-1}(\sigma')\\
    &=(|p_{m-1}|-|p_{m-2}|)\p{1-\p{\sum_{i'<m-1}X(p_{i'})}^n}-
    (|p_m|-|p_{m-2}|)\p{1-\p{\sum_{i'<m-1}X(p_{i'})}^n}~+\\
    &(|p_m|+\alpha-|p_{m-1}|)\p{1-(1-X(p_m))^n}-
    (|p_{m-1}|+\alpha-|p_m|)\p{1-(1-X(p_m))^n}\\
    &=-(|p_m|-|p_{m-1}|)\p{1-\p{\sum_{i'<m-1}X(p_{i'})}^n}+
    2(|p_m|-|p_{m-1}|)\p{1-(1-X(p_m))^n}\\
    &=(|p_m|-|p_{m-1}|)\p{2\p{1-(1-X(p_m))^n}-\p{1-\p{\sum_{i'<m-1}X(p_{i'})}^n}}\\
    &=(|p_m|-|p_{m-1}|)\p{1+\p{\sum_{i'<m-1}X(p_{i'})}^n-2(1-X(p_m))^n}
\end{align*}
Let $x=\sum_{i'<m-1}X(p_{i'}),y=X(p_{m-1}),z=X(p_m)$. Because $x+y+z=1$, this becomes
\begin{align*}
    (|p_m|-|p_{m-1}|)\p{1-(1-z)^n-(x+y)^n+x^n}
\end{align*}
and the proof of the previous case completes the proof of this case, establishing the lemma.
\qed\end{proof}

\end{document}